\documentclass[a4paper,UKenglish,cleveref, autoref, thm-restate, final]{lipics-v2021}
\usepackage{color,xspace}
\usepackage[ruled,linesnumbered]{algorithm2e}
\usepackage{float}
\usepackage{hyperref}
\usepackage[T1]{fontenc}
\usepackage{graphicx}
\usepackage{cite}
\usepackage{mathtools}
\usepackage{wrapfig}
\usepackage{array}

\usepackage{thmtools}
\usepackage{thm-restate}
\usepackage{cleveref}

\usepackage[appendix=append,bibliography=common]{apxproof}

\usepackage[colorinlistoftodos,textsize=tiny,textwidth=2cm,color=red!25!white]{todonotes}

\newcommand{\ignore}[1]{}

\newcommand{\GWSPD}[1]{WSPD$_G$ with separation factor #1}

\newcommand{\GST}[1]{#1-CT}

\newcommand{\diam}[1]{diam(#1)}
\newcommand{\radius}[1]{rad(#1)}
\newcommand{\dist}[2]{dist(#1, #2)} 
\newcommand{\gdist}[2]{dist_{G}(#1,#2)}
\newcommand{\gdub}[1]{\text{dub}_G(#1)}
\newcommand{\rel}[1]{\text{REL($#1$)}}

\newcommand{\R}{\mathbb{R}}

\usepackage{xspace}
\newcommand{\etal}{\emph{et~al.}\xspace}

\hideLIPIcs

\title{A WSPD, Separator and Small Tree Cover for \texorpdfstring{$c$}{c}-packed Graphs} 

\author{Lindsey Deryckere}{School of Computer Science, The University of Sydney, Australia}{lindsey.deryckere@sydney.edu.au}{}{}

\author{Joachim Gudmundsson}{School of Computer Science, The University of Sydney, Australia}{joachim.gudmundsson@sydney.edu.au}{}{}

\author{André van Renssen}{School of Computer Science, The University of Sydney, Australia}{andre.vanrenssen@sydney.edu.au}{}{}

\author{Yuan Sha}{School of Computer Science, The University of Sydney, Australia}{ysha3185@uni.sydney.edu.au}{}{}

\author{Sampson Wong}{Department of Computer Science, The University of Copenhagen, Denmark}{sawo@di.ku.dk}{}{}

\authorrunning{L. Deryckere, J. Gudmundsson, A. van Renssen, Y. Sha and S. Wong}

\Copyright{Lindsey Deryckere, Joachim Gudmundsson, A. van Renssen, Y. Sha and S. Wong}


\nolinenumbers

\funding{This research was partially funded by the Australian Government through the Australian Research Council (project number DP240101353).}

\begin{document}
\maketitle

\begin{abstract}
The $c$-packedness property, proposed in 2010, is a geometric property that captures the spatial distribution of a set of edges. Despite the recent interest in $c$-packedness, its utility has so far been limited to Fr\'echet distance problems. An open problem is whether a wider variety of algorithmic and data structure problems can be solved efficiently under the~$c$-packedness assumption, and more specifically, on~$c$-packed graphs.

In this paper, we prove two fundamental properties of $c$-packed graphs: that there exists a linear-size well-separated pair decomposition under the graph metric, and there exists a constant size balanced separator. We then apply these fundamental properties to obtain a small tree cover for the metric space and distance oracles under the shortest path metric. In particular, we obtain a tree cover of constant size, an exact distance oracle of near-linear size and an approximate distance oracle of linear size.
\end{abstract}

\section{Introduction}

The study of graphs and their properties is a cornerstone of theoretical computer science. A wide variety of graph properties have been proposed in the literature, such as planarity~\cite{Planar-separator-theorem_LiptonTarjan79}, treewidth~\cite{bertele1972nonserial} and doubling dimension~\cite{DBLP:conf/focs/GuptaKL03}. By assuming these graph properties, one can often obtain better algorithmic or data structure solutions to graph theoretic problems.

The $c$-packedness property~\cite{DBLP:journals/dcg/DriemelHW12}, proposed in 2010, is a geometric property that captures the spatial distribution of the edges in a graph. A graph is $c$-packed if, for any positive real $r$ and any ball of radius $r$, the length of the edges contained in the ball is at most $c \cdot r$. Driemel, Har-Peled and Wenk~\cite{DBLP:journals/dcg/DriemelHW12} introduced the $c$-packedness property for polygonal curves, and showed that one can compute the Fr\'echet distance between a pair of $c$-packed curves in near-linear time. In 2013, Gudmundsson and Smid~\cite{DBLP:journals/comgeo/GudmundssonS15} adapted the $c$-packedness definition to graphs, and proposed a Fr\'echet distance data structure on $c$-packed trees with long edges that decides if there is a path in the tree with small Fr\'echet distance to a query curve. In 2023, Gudmundsson, Seybold and Wong~\cite{DBLP:conf/soda/GudmundssonSW23} generalised the result of~\cite{DBLP:journals/comgeo/GudmundssonS15} by proposing a Fr\'echet distance data structure for all $c$-packed graphs. 

Despite the recent interest in $c$-packedness, its utility so far has been limited to Fr\'echet distance problems. Moreover, the fundamental properties of $c$-packed graphs are not well studied, thereby limiting the number of problems that can be solved efficiently on $c$-packed graphs. An open problem is whether $c$-packed graphs have applications beyond Fr\'echet distance problems.

\subsection{Our Results}

We show that $c$-packed graphs lie in the intersection of two important graph classes, that is, doubling metrics and bounded treewidth graphs. Using the properties of doubling metrics and bounded treewidth graphs, one can obtain, for $c$-packed graphs, a linear-size well-separated pair decomposition, a linear-size exact distance oracle, a linear-size approximate distance oracle, and a constant-size tree cover. However, these constructions have processing times that are either $(i)$ randomised, $(ii)$ depend on the spread of the $c$-packed metric, or $(iii)$ are exponential in the treewidth. See Table~\ref{table:results}, Column~3\footnote{A reviewer pointed us to the tree-like properties of the graph and the possibility of adapting the work by Chazelle~\cite{Chazelle87a}. However, such adaptations would be likely to incur an exponential dependency on $c$.}.

We provide the first deterministic constructions that are independent of the spread of the $c$-packed metric, for the aforementioned structures. Moreover, our preprocessing times and data structure sizes are polynomial in both~$c$ and~$\varepsilon$, whereas previous constructions are not.

We summarise the main results of our paper. First, we show that any $c$-packed graph in $\mathbb R^d$ admits a well-separated pair decomposition (WSPD) of size $O((c^3/\varepsilon) \cdot n)$, where $1/\varepsilon$ is the separation constant of the WSPD. Note that we avoid the $(1/\varepsilon)^d$ factor that appears in the sizes of many other WSPD constructions~\cite{DBLP:journals/jacm/CallahanK95,Hierarchical_Nets_WSPD_randomised}. Then, we show that any $c$-packed graph in~$\mathbb R^d$ admits an $O(c)$-size separator. We use this separator to show that $c$-packed graphs have $O(c)$ treewidth, and admits an exact distance oracle (EDO) of size $O(cn \log n)$. Finally, we combine our WSPD and the EDO to construct a $(1+\varepsilon)$ distortion tree cover with $O(c^{2d+2}/ \varepsilon^{d})$ trees. Our tree cover implies an approximate distance oracle (ADO) of size $O((c^{2d+2} /\varepsilon^{d}) \cdot n)$ and $O(c^{2d+2} /\varepsilon^{d})$ query time. We summarise our results in Table~\ref{table:results}.

\begin{table}[ht]
\renewcommand{\arraystretch}{1.2}
\hspace*{-2.75cm}
\begin{minipage}{\textwidth}
\begin{tabular}{|c||c|c|c|c||c|c|c|}
    \hline
    & \multicolumn{4}{c||}{Previous} & \multicolumn{3}{c|}{New}  \\
    \hline
    &  & Preprocessing & Size & Source & 
    Preprocessing & Size & Source
    \\
    \hline
    \multirow{2}{*}{WSPD}
    & \multirow{2}{*}{$dd$}
    & $(1/\varepsilon)^{O(\log c)} \cdot n \log (\Delta n)$ 
     & \multirow{2}{*}{$(1/\varepsilon)^{O(\log c)} \cdot n$}
    & \cite{WSPD_doubling_dim_aspect_ratio_deterministic}
    & \multirow{2}{*}{$(c^3/\varepsilon) \cdot n \log n$}
    & \multirow{2}{*}{$(c^3/\varepsilon) \cdot n$}
    & \multirow{2}{*}{Thm~\ref{thm:GWSPD}}
    \\
    \cline{3-3} \cline{5-5}
    & & Rand. $(1/\varepsilon)^{O(\log c)} \cdot n \log n$ & 
    &\cite{Hierarchical_Nets_WSPD_randomised}
    & & &
    \\
    \hline
    EDO 
    & $tw$
    & $2^{O(c^3)} n$
    & $c^2 n$
    & \cite{ShortestPathQuery_Digraph_Chaudhuri&Zaroliagis}+\cite{DvorakN19_separator_treewidth}
    & $c^2 n \log^2 n$
    & $c n \log n$
    & Thm~\ref{thm:exact-distance-oracle}
    \\
    \hline
    Tree cover
    & $dd$
    & $(1/\varepsilon)^{O(\log c)} \cdot n \log^2 (\Delta n)$ 
     & $(1/\varepsilon)^{O(\log c)}$
    & \cite{tree-cover-doubling-metrics_BartalFN-JCSS2022} + \cite{WSPD_doubling_dim_aspect_ratio_deterministic}
    & {$(c^{2d+6}/\varepsilon^{d+1}) \cdot n \log n$}
    & {$c^{2d+2} /\varepsilon^{d}$}
    & {Thm~\ref{thm:tree-cover-c-packed}}
    \\
    \hline
   ADO
   & $dd$
    & $(1/\varepsilon)^{O(\log c)} \cdot n \log^2 (\Delta n)$ 
     & $(1/\varepsilon)^{O(\log c)} \cdot n$
    & \cite{tree-cover-doubling-metrics_BartalFN-JCSS2022} + \cite{WSPD_doubling_dim_aspect_ratio_deterministic}
    & $(c^{2d+6}/\varepsilon^{d+1}) \cdot n \log n$
    & $(c^{2d+2} /\varepsilon^{d}) \cdot n$
    & Thm~\ref{col:tree-cover-apprx-do}
    \\
    \hline
\end{tabular}
\end{minipage}
\vspace{10pt}
\caption{In the table, $c$ = the $c$-packedness value, $\varepsilon$ = either an $\varepsilon^{-1}$ separation constant or a~$(1+\varepsilon)$ approximation ratio, Rand. = randomised algorithm, $\Delta$ = spread of the doubling metric, $d$ =  dimension of the Euclidean space, $dd$ = previous results using doubling dimension, $tw$ = previous results using treewidth.}
\label{table:results}
\end{table}

\subsection{Related work}

The $c$-packedness property is a popular model for realistic curves. A wide range of Fr\'echet distance problems have been studied on $c$-packed curves, including map matching~\cite{Chen_map_match_lowdensegraph_cpackedcurve}, the mean curve~\cite{DBLP:journals/talg/Har-PeledR14}, the shortcut Fr\'echet distance~\cite{Driemel2013}, subtrajectory clustering~\cite{DBLP:conf/esa/BruningCD22,DBLP:conf/isaac/GudmundssonHRW23} and the approximate nearest neighbour data structure~\cite{DBLP:journals/cgt/ConradiDK24}. However, $c$-packed graphs are less well understood~\cite{DBLP:conf/soda/GudmundssonSW23,DBLP:journals/comgeo/GudmundssonS15}, despite often being used as a vital stepping stone towards a related property called $\lambda$-low density, which was also introduced by Driemel and Har-Peled~\cite{Driemel2013}. Low density graphs have been studied in map matching~\cite{Chen_map_match_lowdensegraph_cpackedcurve,buchin2023map}. The well-separated pair decomposition of~\cite{gudmundsson2024well} has size polynomial in $\lambda$ and $\varepsilon$, but its disadvantage over those stated in Table~\ref{table:results} is that it has size $O(n \log n)$. 

Well-Separated Pair Decompositions (WSPD) are used for compact representation of the quadratic distances between pairs of points in the metric. For metrics that allow for a sub-quadtratic size WSPD, they have therefore been used as fundamental tools to approximate solutions to a range of proximity problems that require looking at the distances between all pairs of points, such as nearest neighbour, diameter, stretch and minimum spanning tree. Not all metrics allow for a WSPD of subquadratic size, an example of which is the metric induced by a star tree with unit weight on all edges. For this metric, any WSPD requires a quadratic number of pairs. For a point set in $\mathbb R^d$, where $d$ is considered a constant, Callahan and Kosaraju~\cite{DBLP:journals/jacm/CallahanK95} showed that there exists a WSPD with separation factor $\sigma$ of size $O(\sigma^d n)$ that can be computed in $O(n \log n + \sigma^d n)$ time. In contrast to this, we show that for $c$-packed graphs, the size of the WSPD is not exponential in $d$, while maintaining that the size is linear. For graphs with bounded doubling dimension ($dim$), Har-Peled and Mendel~\cite{Hierarchical_Nets_WSPD_randomised} designed a $O(2^{O(dim)} n \log n + n\varepsilon^{-O(dim)})$ expected time randomised algorithm to construct a WSPD of linear size with logarithmic query time. They also designed a deterministic construction which incurs a logarithmic dependency on the aspect ratio of the metric space. In this paper we show that a $c$-packed graph has doubling dimension $O(\log c)$ (see Lemma~\ref{lem:doublingProof}). However, in contrast to previous results, our deterministic construction of the WSPD is, to the best of our knowledge, the first that does not depend on the aspect ratio of the metric space and does not incur any factors exponential in the dimension, assuming constant dimension. 

Distance oracles are shortest path data structures for graphs. For planar graphs, Lipton and Tarjan~\cite{Planar-separator-theorem_LiptonTarjan79,planar-separator-applications_LiptonT80} proved the planar separator theorem, which states that any planar graph with~$n$ vertices has a balanced separator of size~$\sqrt n$. Using the separator to build a separator hierarchy, they constructed an exact distance oracle of size $O(n \sqrt n)$. Thorup~\cite{DBLP:journals/jacm/Thorup04} and Klein~\cite{DBLP:conf/soda/Klein02} independently presented $(1+\varepsilon)$-approximate distance oracles for planar graphs of size $O(n \log n)$, for any constant $\varepsilon > 0$. Subsequent works~\cite{DBLP:conf/icalp/KawarabayashiKS11,DBLP:conf/soda/KawarabayashiST13,DBLP:conf/soda/Wulff-Nilsen16,DBLP:conf/focs/LeW21} improved the preprocessing time, space, and query time. Dvořák and Norin~\cite{DvorakN19_separator_treewidth} showed that if a graph admits a small size balanced separator, it also has small treewidth. Combined with our separator results this implies that $c$-packed graphs have treewidth $O(c)$. Chaudhuri and Zaroliagis\cite{ShortestPathQuery_Digraph_Chaudhuri&Zaroliagis} designed an exact distance oracle whose preprocessing time is single exponential in the treewidth of the graph. In contrast to these results, our algorithms do not incur any terms exponential in $c$.

Approximate distance oracles have been studied on non-planar graphs. Thorup and Zwick~\cite{ApprxDistanceOracle_ThorupZ05} constructed a $(2k-1)$-approximate distance oracle of $O(kn^{1+1/k})$ size on any general graph, and showed that the size bound is optimal under the Erd\"os girth conjecture. Gudmundsson, Levcopoulos, Narasimhan and Smid~\cite{DBLP:journals/talg/GudmundssonLNS08} provided a $(1+\varepsilon)$-approximate distance oracle of size $O(n \log n)$ on any $t$-spanner, assuming $\varepsilon, t > 0$ are constants. Gao and Zhang~\cite{DBLP:journals/siamcomp/GaoZ05} constructed a $(1+\varepsilon)$-approximate distance oracle of size $O(n \log n)$ for any unit-disk graph, assuming $\varepsilon > 0$ is a constant.

Balanced separators of sublinear size have been found for planar graphs~\cite{Planar-separator-theorem_LiptonTarjan79}, graphs of bounded genus~\cite{Separator-bounded-genus_GilbertHT84}, graphs that exclude a fixed minor~\cite{Separator-exclude-minor_AlonST90}, and graphs that are the 1-skeletons of simplicial complexes in 3-dimensions~\cite{Separator-simplicial-complex_MillerT90}. They have been used as a fundamental tool in devising efficient algorithms for graphs~\cite{Separator-exclude-minor_AlonST90, Separator-bounded-genus_GilbertHT84, planar-separator-applications_LiptonT80} and in numerical analysis~\cite{incomplete-nested-dissection_KyngPSZ18, Separator-simplicial-complex_MillerT90}. Randomized balanced separators can be found in expected linear time for $c$-packed graphs~\cite{DBLP:journals/siamcomp/Har-PeledQ17}. 

Metric embeddings approximate harder metric spaces by simpler, well-structured metric spaces, such as tree metrics. A tree cover for a finite metric space is a small number of trees such that the distance between any two points in the metric is preserved with low distortion in at least one of the trees. Tree covers with $(1+\epsilon)$-distortion and a constant number of trees have been found for Euclidean metrics~\cite{dumbbell-theorem-Euclidean-stoc95}, doubling metrics~\cite{tree-cover-doubling-metrics_BartalFN-JCSS2022} and planar metrics~\cite{tree-cover-planar-metrics_ChangEtal-focs2023}.

\section{Preliminaries}
Let $G(V,E)$ be a geometric graph in $\mathbb{R}^d$ consisting of the point set $V$ and edge set $E$. In this paper, we consider graphs in a fixed $d$-dimensional space that satisfy $c$-packedness. 
We denote the graph distance between two nodes $u,v \in V$ as $\gdist{u}{v}$, and their Euclidean distance as $\dist{u}{v}$. 
We first define a Well-Separated Pair Decomposition~\cite{DBLP:journals/jacm/CallahanK95} for Euclidean point sets, followed by its counterpart for geometric graphs.

\begin{definition}[Geometric Well-Separated Pair] Let $\sigma$ be a real positive number, and let $A$ and $B$ be two finite sets of points in $\mathbb{R}^d$. We say that $A$ and $B$ are \emph{well-separated with respect to $\sigma$} if the distance between the bounding boxes $C_A$ and $C_B$ of $A$ and $B$, respectively, is at least $\sigma \cdot \max(\radius{C_A}, \radius{C_B})$.
\end{definition}
Here $\radius{C_A}$ refers to the radius of $C_A$, which is half its diameter.

\begin{definition}[Geometric Well-Separated Pair Decomposition (WSPD)]
    Let $S$ be a set of $n$ points in $\mathbb{R}^d$, and let $\sigma$ be a real positive number. A \emph{well-separated pair decomposition} (WSPD) for $S$, with respect to $\sigma$, is a sequence $\{A_1, B_1\}, \ldots \{A_m,B_m\}$ of pairs of nonempty subsets of $S$, for some integer $m$, such that (i) for each $i$ with $1\leq i\leq m$, $A_i$ and $B_i$ are well separated w.r.t. $\sigma$, and (ii) for any two distinct points $p$ and $q$ of $S$, there is exactly one index $i$ with $1\leq i \leq m$, such that $p\in A_i$ and $q\in B_i$, or vice versa.
\end{definition}

The notion of well-separated pairs and WSPD can easily be extended to graphs. For the graph version we will simply replace all the geometric distances with the distances in the graph. We will refer to a WSPD in a graph $G$ as WSPD$_G$.

\section{Technical Overview}
\label{section:technical_overview}

In Section~\ref{technical_overview:GWSPD}, we describe our deterministic WSPD$_G$ construction. In Section~\ref{technical_overview:SmallSeparator}, we discuss our separator theorem. 
In Section~\ref{technical_overview:Dist_oracles}, we describe our exact distance oracle of~$O(n \log n)$ size.
Finally, in Section~\ref{ssec:overview-tree-cover} we show how to use the WSPD$_G$ and the exact distance oracle to construct a $(1+\varepsilon)$-distortion tree cover of size $O(c^{2d+2}(\frac{1}{\varepsilon})^d)$ for $c$-packed metrics.
 
\subsection{A Well-Separated Pair Decomposition for \texorpdfstring{$c$}{c}-packed Graphs}
\label{technical_overview:GWSPD}
  
Split trees~\cite{DBLP:journals/jacm/CallahanK95} or quadtrees~\cite{spanner-book} are commonly used in the construction of a geometric WSPD of a point set. An essential property of these is that the maximum Euclidean distance between two points contained in the same cell is always bounded by a function of the diameter of the cell. We construct a tree that fulfills a similar purpose to the trees above but for graph distances between points. We call this new type of tree a \emph{$\delta$-connected tree} ($\delta$-CT). Each cell, corresponding to a cube $s$, of the $\delta$-connected tree is a $\delta$-connected set, meaning that points contained in the cell are within a graph distance of at most $\delta \cdot \diam s$ from one another. To construct the $\delta$-CT, we use a bottom up approach. The leaves of the compressed quadtree are already $\delta$-connected sets. At higher levels of the compressed quadtree, we consider the $\delta$-connected sets of its children, and merge together pairs of previously $\delta$-connected sets that are also a $\delta$-connected set in the higher level. To obtain an efficient running time, we make two observations. First, when computing the $\delta$-connected sets for a higher level, it suffices to maintain a vertex representative for each $\delta$-connected set of the lower level. Second, to check if a pair of sets are $\delta$-connected, it suffices to check whether their representatives are path-connected in the cube centred at the cell but with double its radius. 

To upper bound the graph diameter of the $\delta$-connected set in each cell of the $\delta$-CT we compute the length of intersection of edges with the cell and the $3^d$ surrounding cells in a canonical grid. To do this efficiently, we note that not all edges intersecting with a cell can contribute to a $\delta$-connected component inside the cell. We therefore construct a data structure that can be queried for the total length of all edges that can contribute to a $\delta$-connected component contained in a cell.

Combining these ideas obtains the following theorem. For details refer to Section~\ref{subsec:construction}.

\begin{restatable}{thm}{GWSPDthm}
\label{thm:GWSPD}
    Given a $c$-packed graph $G$ in $\mathbb{R}^d$, for fixed $d$, one can construct a WSPD$_G$ with separation factor $\sigma$ of size $O(c^{3} \sigma n)$ in $O(cn \log n + c^3\sigma n)$ time, using $O(cn)$ space.
\end{restatable}

\subsection{A Separator Theorem for \texorpdfstring{$c$}{c}-packed Graphs}
\label{technical_overview:SmallSeparator}

We prove that every $c$-packed graph admits a balanced vertex separator of size~$O(c)$. We start with the ring separator of Har-Peled and Mendel~\cite{Hierarchical_Nets_WSPD_randomised}, which states that for a point set in $\mathbb R^d$, one can efficiently compute a pair of balls so that $n/2\lambda^3$ of the points are inside the inner ball, and $n/2\lambda^3$ of the points are outside the outer ball, where $\lambda$ is the doubling constant of $\mathbb{R}^d$. Using the ring separator, we construct a max-flow instance in a similar fashion to Gudmundsson et al.~\cite{DBLP:conf/soda/GudmundssonSW23} to locate a cut of size~$O(c)$. This cut $(1 - 1/2\lambda^3)$-separates the graph, in that it separates the graph into two components each with at most $n \cdot (1-1/2\lambda^3)$ points. We obtain the following theorem, for details refer to Section~\ref{sec:SmallSeparator}.

\begin{restatable}{thm}{separator}
\label{thm:separator-technical-overview}
        Given a $c$-packed graph $G$ in $\mathbb R^d$, where $d$ is fixed, with $n$ vertices, one can find a separator of size $O(c)$ that $(1-\frac{1}{2\lambda^3})$-separates $G$, in $O(c^2n)$ time.
\end{restatable}

\subsection{Distance Oracles for \texorpdfstring{$c$}{c}-packed Graphs}
\label{technical_overview:Dist_oracles}

We can combine our separator with standard techniques to construct an exact distance oracle. We recursively apply the separator theorem to obtain a separator heirarchy~\cite{planar-separator-applications_LiptonT80}, which is a balanced binary tree decomposition. Each internal node stores a separator and a shortest path tree rooted at each vertex in the separator. At query time, we perform a lowest common ancestor query to find the appropriate internal node and use the shortest path trees to compute the exact shortest path distance. Putting this together, we obtain Theorem~\ref{thm:exact-distance-oracle}. For details refer to Section~\ref{sec:Dist_oracles}. 

\begin{restatable}{theorem}{exactdistoracle}
\label{thm:exact-distance-oracle}
    Given any $c$-packed graph $G$ with $n$ vertices, using $O(c^2n\log n+cn\log^2 n)$ preprocessing time and $O(cn\log n)$ space, a distance query between any two vertices in $G$ can be answered in $O(c\log n)$ time.
\end{restatable}

\subsection{A Small Tree Cover for \texorpdfstring{$c$}{c}-packed Graphs}
\label{ssec:overview-tree-cover}

Our approach follows that of the celebrated ``Dumbbell Theorem''~\cite{dumbbell-theorem-Euclidean-stoc95} for Euclidean metrics. For $c$-packed graphs, we construct a linear number of dumbbells from the WSPD$_G$ in Corollary~\ref{thm:GWSPD-graph-distance} which is constructed using the graph distances. The dumbbells are partitioned into groups, each of which satisfies the \emph{length-group property} and the \emph{empty-region property} with respect to the graph distance. The $c$-packedness property enables us to partition the dumbbells into a small number (depending only on the packedness value $c$ and the separation ratio $\sigma$ of WSPD$_G$) of groups, each of which satisfies the empty-region property. The main difficulty is proving the packing lemmas required for establishing the empty-region property. A dumbbell tree, which connects the dumbbells in a group hierarchically, is built for each group of dumbbells. The $c$-packedness property and the $c$-CT also enable us to do range searching and efficiently build the dumbbell trees. The dumbbell trees together constitute a tree cover for the $c$-packed metric. We obtain the following theorem. For details refer to~Section~\ref{sec:tree_cover}.

\begin{restatable}{theorem}{treecovermain}
\label{thm:tree-cover-c-packed}
    Given a $c$-packed graph $G$ in $\R^d$ of fixed $d$ and any $0<\varepsilon<1$. In $O(n\log n(c^{2d+6}(\frac{1}{\varepsilon})^{d+1}+c\log n))$ time, one can construct $O(c^{2d+2}(\frac{1}{\varepsilon})^d)$ dumbbell trees, each of size $O(n)$, such that the dumbbell trees constitute a $(1+\varepsilon)$-distortion tree cover for the graph metric induced by $G$.
\end{restatable}

The tree cover of Theorem~\ref{thm:tree-cover-c-packed} immediately implies a $(1+\varepsilon)$-approximate distance oracle for the $c$-packed metric.
\begin{restatable}{corollary}{apprxDOmain}\label{col:tree-cover-apprx-do}
    Given a $c$-packed graph $G$ in $\R^d$ of fixed $d$, and any $0<\varepsilon<1$, one can preprocess it in $O(n\log n(c^{2d+6}(\frac{1}{\varepsilon})^{d+1}+c\log n))$ time, using $O(c^{2d+2}(\frac{1}{\varepsilon})^dn)$ space, to answer a $(1+\varepsilon)$-approximate distance query between any two vertices in $G$ in $O(c^{2d+2}(\frac{1}{\varepsilon})^d)$ time.
\end{restatable}

\begin{toappendix}
\section{Doubling dimension of $c$-packed graphs}
\begin{lemma}
\label{lem:doublingProof}
    The doubling constant of a $c$-packed graph $G$ is at most $4c+1$.
\end{lemma}
\begin{proof}
    Let $G$ be a $c$-packed graph in $\R^d$. Let $(X,d_G)$ denote the metric induced by the graph distance in $G$ where the elements of $X$ are the points on $G$ (including all the points on the edges). Let $p$ be a point on $G$ and let $\mathbf{B}_G(p,r)$ denote the metric ball in $X$ with center $p$ and radius $r$. In the following, we prove that $\mathbf{B}_G(p,r)$ can be covered by at most $4c+1$ metric balls of radius $r/2$, thus proving the lemma.

    Let $SPT(p)$ be the shortest path tree of $p$ on $G$, including shortest paths from $p$ to all other points on $G$. Let $SPT(p,r)$ be the subtree of $SPT(p)$ which only includes points within distance $r$ to $p$. Thus $SPT(p,r)$ contains the shortest paths from $p$ to all points in $\mathbf{B}_G(p,r)$. The longest path in $SPT(p,r)$ from $p$ to a leaf point is called the \emph{axis} of $SPT(p,r)$. A subtree of $SPT(p,r)$ whose root is on the axis of $SPT(p,r)$ is called a \emph{branch} of $SPT(p,r)$. Similarly, the longest path in the subtree from its root to a leaf point is called the axis of the subtree. In Figure~\ref{fig:doubling-dimension}, the axis of $SPT(p,r)$ is the path from $p$ to $d$. The subtree rooted at $p_1$ is a branch of $SPT(p,r)$ whose axis is $p_1q_1$.

    \begin{figure}[H]
	\begin{center}
		\includegraphics[width=.85\textwidth]{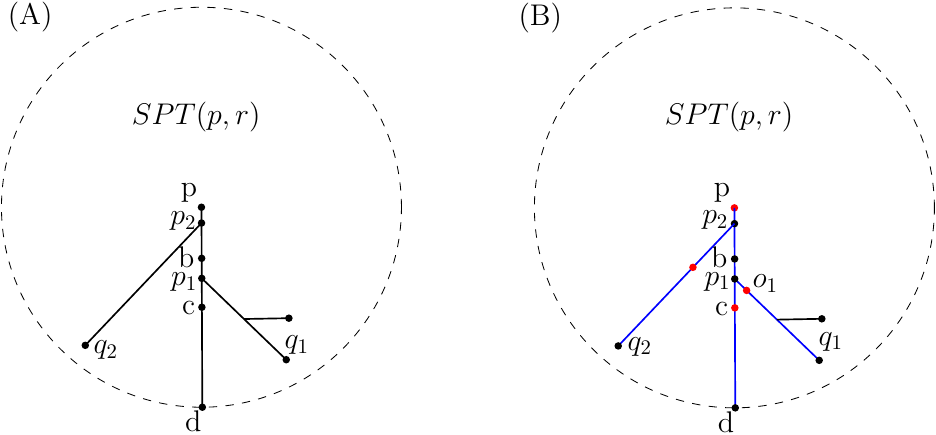}
		\caption{(a) The $SPT(p,r)$. The path from $p$ to $d$ is the axis of $SPT(p,r)$, the subtree rooted at $p_1$ is a branch of $SPT(p,r)$. (b) The chosen centers are drawn in red.}
		\label{fig:doubling-dimension}
	\end{center}
    \end{figure} 
    
    With a bit abuse of notation, let $pd$ denote the path from $p$ to $d$. Without loss of generality, assume that $|pd|=r$ ($|pd|$ is the length of $pd$). Let $b$ be the point on the axis $pd$ such that $|pb|=\frac{r}{4}$ and let $c$ be the point on $pd$ such that $|pc|=\frac{r}{2}$. Next we choose the centers of the balls of radius $r/2$ which together cover $\mathbf{B}_G(p,r)$.

    Choose $c$ as a center. For each branch whose root is on path $bc$ and whose axis has length at least $r/4$, choose centers on it recursively as follows. Let $Br$ be such a branch. If the length of its axis is at most $r/2$, choose its root as a center and stop. Otherwise the length of its axis is between $r/2$ and $3r/4$. Let $p_iq_i$ be its axis where $p_i$ is the root of the branch. Choose $o_i$ on $p_iq_i$ where $|o_iq_i|=r/2$ as a center. For each branch of $Br$ whose root is on $p_io_i$ and whose axis has length at least $r/4$, choose centers on the branch recursively as above. For example in Figure~\ref{fig:doubling-dimension}(b), the branch rooted at $p_1$ has axis $p_1q_1$ whose length is longer than $r/2$. Choose $o_1$ on $p_1q_1$ as a center where $|o_1q_1|=r/2$. For each branch whose root is on $p_1o_1$ and whose axis has length at least $r/4$, choose centers on the branch using the above recursive procedure.

    Then choose $p$ as a center. For each branch whose root is on $pb$ and whose axis has length at least $r/4$, choose centers on it recursively as follows. Let $BH$ be such a branch. If the length of its axis is at most $r/2$, choose its root as a center and stop. Otherwise the length of its axis is between $r/2$ and $r$. Let $p_jq_j$ be its axis where $p_j$ is its root. If $|p_jq_j|\leq 3r/4$, choose the centers on $BH$, using the recursive procedure in the last paragraph (for the branch whose axis is $p_iq_i$). Otherwise $3r/4 < |p_jq_j|$ and let $t_j$ be the point on $p_jq_j$ where $t_jq_j=3r/4$. Choose $o_j$ on $p_jq_j$ where $|o_jq_j|=r/2$ as a center. For each branch of $BH$ whose root is on $t_jo_j$ and whose axis has length at least $r/4$ (and at most $3r/4$), choose centers on it using the recursive procedure in the last paragraph (for the branch whose axis is $p_iq_i$). For each branch of $BH$ whose root is on $p_jt_j$ and whose axis has length as least $r/4$, choose centers on it using the recursive procedure in this paragraph.
    
    We can prove that the algorithm chooses centers such that the balls of radius $r/2$ centered at the centers cover all the points in $\mathbf{B}_G(p,r)$. The ball (of radius $r/2$) centered at $c$ covers points on path $cd$ and all the branches of $SPT(p,r)$ whose roots are on $cd$. The ball centered at $c$ also covers points on $bc$ and all the branches of $SPT(p,r)$ whose roots are on $bc$ and whose axis have length at most $r/4$. For any branch of $SPT(p,r)$ whose root is on $bc$ and whose axis has length at least $r/4$, if its axis has length at most $r/2$, the branch is covered by the ball centered at its root, otherwise its axis has length between $r/2$ and $3r/4$ and the branch is covered by the balls centered at the centers chosen on the branch. Using a similar argument, any branch of $SPT(p,r)$ whose center is on $pb$ and whose axis has length at least $r/4$, the branch is covered by the balls centered at the centers chosen on the branch. The ball centered at $p$ covers points on $pb$ and all the branches whose centers are on $pb$ and whose axis have length at most $r/4$. Therefore all the points in $\mathbf{B}_G(p,r)$, are covered.

    Finally we bound the number of centers chosen by the algorithm. Observe that except the center $p$, the center $c$ and the centers chosen on the branches, each corresponds to a disjoint path on $SPT(p,r)$ with length at least $r/4$. For example in Figure~\ref{fig:doubling-dimension}(b), $c$ corresponds to path $cd$ and $o_1$ corresponds to path $o_1q_1$. By the $c$-packedness property, the total edge length of $SPT(p,r)$ is at most $cr$. Therefore the number of centers, including $c$ and the centers chosen on the branches, is at most $4c$. Plus the center $p$, the total number of centers chosen is at most $4c+1$.
\end{proof}
\end{toappendix}
\section{A Well-Separated Pair Decomposition for \texorpdfstring{$c$}{c}-packed Graphs}
\label{sec:GWSPD}
Given a $c$-packed graph $G(V,E)$ in $d$ dimensions, where $d$ is fixed, and a positive constant $\sigma$, we show how to deterministically construct a linear-size \GWSPD{$\sigma$} (Section~\ref{subsec:construction}).
Section~\ref{subsec:sec4_notation} introduces some important terminology. In Section~\ref{subsec:construction} we show that one can construct such a well-separated pair decomposition in $O(n \log n + c^3\sigma n)$ time.

\subsection{Notation and Preliminaries}
\label{subsec:sec4_notation}
Given a geometric graph $G(V,E)$ and cube $s$ in $\mathbb{R}^d$, we define the diameter of $s$ to be the length of a diagonal and denote it as $\diam{s}$. We define the radius of $s$ to be half the diameter and denote it as $\radius{s}$. We define $V(s)$ to be the subset of vertices of $V$ within $s$. We will need the following definitions.

\begin{definition}[$\delta$-Connected Set]
    \label{defn:delta_connected_set}
    Given a geometric graph $G(V,E)$ and a cube $s$ in $\mathbb{R}^d$. Let $s^+$ be the concentric cube with twice the diameter. Two vertices $u,v\in V(s)$ are \emph{$\delta$-connected} if there is a path between $u$ and $v$ that lies within $s^+$ and has length at most $\delta \cdot \diam{s}$. We say that a set of vertices $C \subseteq V(s)$ is a $\delta$-connected set with respect to $s$ if all pairs of vertices in $C$ are $\delta$-connected, and no vertex in $V(s)\setminus C$ is $\delta$-connected to a vertex in $C$.
\end{definition}

\begin{definition}[Partition into $\delta$-Connected Sets]
    Let $s$ be a cube in $\mathbb{R}^d$. A partition $\Psi_{\delta}(V(s)) = \{C_1,... C_k\}$ of $V(s)$ into $\delta$-connected sets is a set of $k$ disjoint subsets of $V(s)$ that satisfies the following properties:
    \begin{enumerate}
        \item $\bigcup_{1\leq i \leq k}C_i = V(s)$.
        \item For all $i$ s.t. $1\leq i \leq k$, $C_i$ is a $\delta$-connected set with respect to $s$.
    \end{enumerate}
\end{definition}

\subsection{Constructing a WSPD for \texorpdfstring{$c$}{c}-packed graphs}
\label{subsec:construction}
Split trees~\cite{DBLP:journals/jacm/CallahanK95} or quadtrees~\cite{spanner-book} are commonly used in the construction of a geometric WSPD of a point set. An essential property of these is that the Euclidean distance between two points contained in the same cell is always bounded by the diameter of the cell. In order to construct a WSPD$_G$ we construct a new type of tree, which we will refer to as a $\delta$-connected tree (\GST{$\delta$}). This tree will satisfy the similar property, but relative to \emph{graph distances} between the points in the graph, rather than their Euclidean distances. The main difference is that cells in a $\delta$-connected tree may overlap. We formally define this tree below. 

\begin{definition}[$\delta$-Connected Tree ($\delta$-CT)]
\label{def:gst}
Given a connected geometric graph $G(V,E)$ in $\mathbb{R}^d$ and a quadtree $Q$ of $V$, a $\delta$-connected tree $T$ of $G$ is a rooted tree, where every node $u$ stores a cube $s$, corresponding to a cell of $Q$, and a representative point of a $\delta$-connected set with respect to $s$. The root of $T$ stores the cube $s$ corresponding to the cell stored in the root of $Q$, and every leaf contains a single point in $V$. In particular, the leaves contained in the subtree rooted at $u$ contain the points in $V$ that are represented by the point stored at $u$.
\end{definition}
Once we have a $\delta$-CT, it remains to upper bound the graph diameter of the $\delta$-connected set represented in each cell. With this upper bound, one can then follow a standard approach to compute a WSPD$_G$ with separation factor $\sigma$ from the $\delta$-CT.

In Subsection~\ref{subsec:c_CT} we show how to construct a $c$-CT, $T$, of a $c$-packed graph $G$. In Subsection~\ref{subsec:upperbounddiam} we then show how to upper bound the diameter of the $c$-connected set represented in each cell of $T$. In Subsection~\ref{subsec:analysisWSPD} we analyze the complexity of the construction.

\subsubsection{Constructing a $c$-Connected Tree}
\label{subsec:c_CT}
The algorithm to compute a $c$-CT takes as input a $c$-packed graph $G$, as well as its corresponding compressed quadtree $Q$. Without loss of generality we assume that the cell representing the root of $Q$ is a unit cube. The \emph{level value} of a node $u$ in $Q$ is said to be $i$ if the cell/cube $s(u)$ stored at $u$ has side length $2^{-i}$. The root of $Q$ has level value $0$.   

We are now ready to construct a $c$-CT, $T$. Initially, every leaf in $Q$ becomes a leaf in $T$, and are then deleted from $Q$. Assume w.l.o.g. that the number of distinct level values in the remaining compressed quadtree $Q^\prime$ is $\ell$ and let $I=\langle i_1, i_2, \ldots, i_{\ell}\rangle$ be the level values sorted in decreasing order. Let $U_{i_j}$ be the set of nodes in $Q'$ having level value $i_j$, $1\leq j \leq \ell$.  

Next, we iteratively process the values in $I$ in decreasing order. For each value $i_j$ in $I$ we will build a graph $H_{i_j}$ from $H_{i_{j-1}}$. Initially we set $H_{i_0}(V_{i_0},E_{i_0})=G$.

Set $H_{i_1}=H_{i_0}$. For each node $u \in U_{i_1}$, let $s^+(u)$ be a concentric cube of twice the diameter of $s(u)$. The algorithm picks an arbitrary point $p$ in $V_{i_1}(s(u))$ and merges all nodes in $V_{i_1}(s(u))$ that are path-connected to $p$ in $H_{i_0}$ within the cube $s^+(u)$. For all edges $(u,v)$, if both $u$ and $v$ have been merged, we remove the edge. If only one endpoint has been merged, we move this endpoint to $p$ and redefine the length of the edge to be the Euclidean distance between this point and $p$. Finally, we remove any parallel edges. Note that all changes are made to $H_{i_1}$, while $H_{i_0}$ stays unaffected. The surviving node $p$ is the representative point in $H_{i_1}$ of the merged set of nodes in $H_{i_0}$. A node is added to $T$ storing $p$ and its children are the nodes in $T$ storing the points that were merged into $p$. The algorithm repeats this process until all the points in $V_{i_1}(s(u))$ have been merged into representative points and the process has been complete for each $u \in U_{i_1}$.  

In the next iteration, we let $H_{i_2} = H_{i_1}$. The nodes in $U_{i_2}$ are now processed using $H_{i_1}$ as input to merge the nodes in $H_{i_2}$. This continues until all the level values in $I$ have been processed. The process for iteration $j$ of the algorithm is visualised in Figure~\ref{fig:GST_Construction}.

\begin{figure}[th]
    \begin{center}
        \includegraphics[scale=0.55]{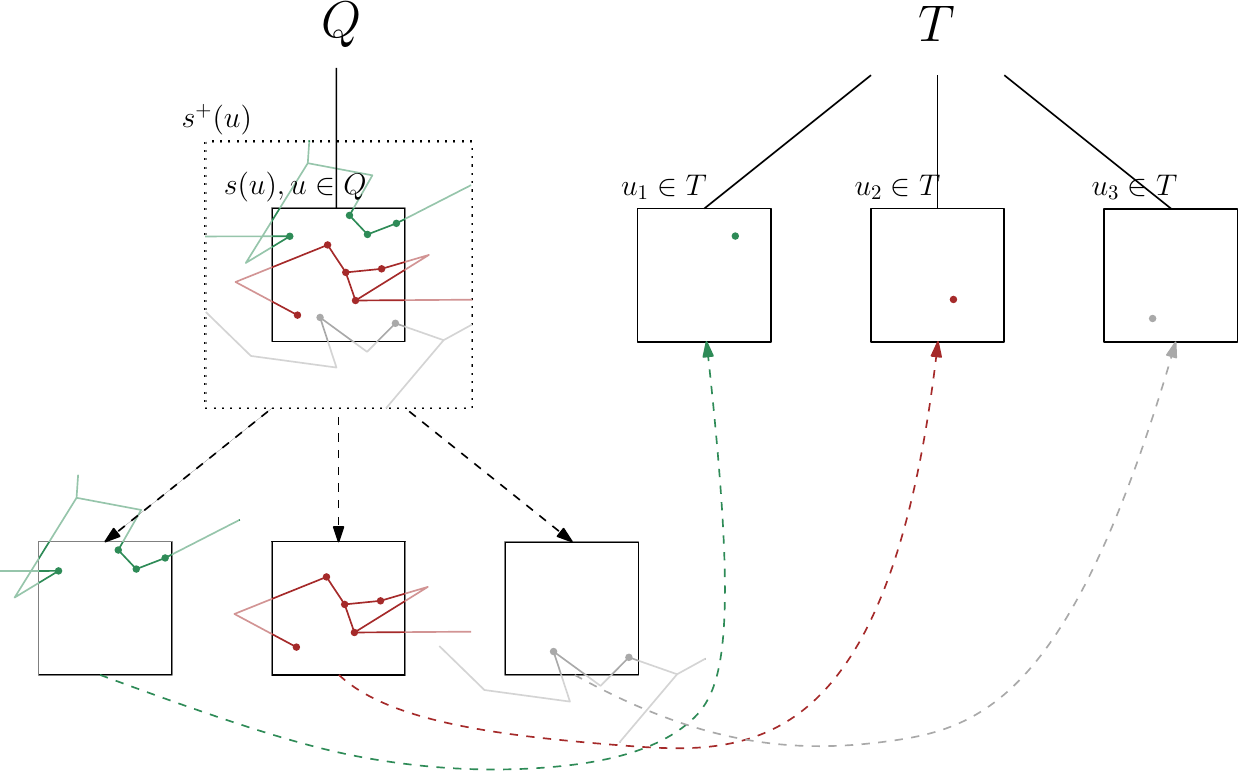}
        \caption{An illustration of iteration $j$ of the algorithm constructing the $c$-connected tree.}
        \label{fig:GST_Construction}
    \end{center}
\end{figure}

\begin{restatable}{lemma}{validGST}
    \label{lem:valid_GST}
    Given a $c$-packed graph $G(V,E)$ the above algorithm produces a valid $c$-connected tree $T$ of $G(V,E)$.
\end{restatable}

\begin{proof}
    To prove the lemma, it is clear from Definition~\ref{def:gst} that it only remains to prove that every node in $T$ is a $c$-connected set in $G$. That is, consider any node $u$ in $T$ and let $p$ and $q$ be any two points stored in the leaves of the subtree rooted at $u$. We will need to prove that $\gdist{p}{q} \leq c \cdot \diam{s(u)}$.
    
    Let $v(p)$ and $v(q)$ be the leaves storing $p$ and $q$ and let $v$ be the lowest common ancestor of $v(p)$ and $v(q)$ in $T$. Let $s(v)$ denote the cube corresponding to node $v$. We claim that the shortest path between $p$ and $q$ in $G$ must lie entirely within a cube $s^+(v)$ of diameter twice the diameter of $s(v)$, centered at $s(v)$. This follows from the fact that every cube and corresponding cube of twice the diameter on the path from $v(p)$ (resp. $v(q)$) to $v$ must lie strictly within $s^+(v)$. From the $c$-packedness of the original graph we deduce that the length of the path between $p$ and $q$ can be at most $2c\cdot \radius{s(v)} = c \cdot \diam{s(v)}$. 

    This allows us to conclude that, since $\diam{s(v)}$ is the smallest diameter of any cell in which $p$ and $q$ appear together, it holds for all cells of $T$ that the original graph distance between any two points contained in the leaves of its subtree, is at most $c$ times the diameter of the cell. We conclude that $T$ is indeed a valid \GST{$c$}.
\end{proof}
\subsubsection{Bounding the Graph Diameter of a Cell in the $c$-CT}
\label{subsec:upperbounddiam}
In order to construct the well-separated pairs, it is necessary to upper bound the graph diameter of the $c$-connected set represented in each cell of the $c$-CT, $T$, constructed above. Let $u$ be a node in $Q$. Denote the corresponding canonical cube as $s(u)$, and a concentric cube of twice its diameter as $s^+(u)$. We call a connected subgraph of $G$ a connected component of a cell $u$ if it is fully contained within $s^+(u)$, with at least one vertex inside $s(u)$. To compute an upper bound on the graph diameter of the $c$-connected set we instead give an upper bound on \emph{any} connected component of $u \in Q$. This upper bound can then be copied across during the construction of $T$. We show how to compute this upper bound below.

Observe that in order for an edge to contribute to a connected component of $u$, its length can be at most twice the diameter of $s(u)$. Let us call the set of edges of length at most twice the diameter, which overlap with $s(u)$, \emph{relevant} edges w.r.t. $s(u)$ (see Figure~\ref{fig:relevant_edges}). 

\begin{figure}[th]
    \begin{center}
        \includegraphics[scale=0.8]{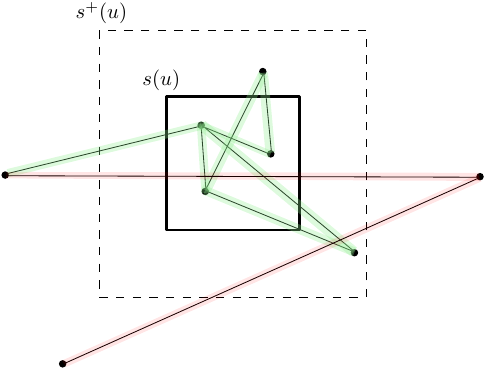}
        \caption{Set of \emph{relevant} edges w.r.t. $s(u)$ (green edges) and \emph{irrelevant} edges (red edges).}
        \label{fig:relevant_edges}
    \end{center}
\end{figure}
Let $N(s(u))$ be the set of $3^d-1$ neighboring canonical cubes on a grid (note that not all members of $N(s(u))$ may correspond to cells in $Q$). Observe that all edges in any connected component associated with $u$ must be a subset of the union of the relevant edges w.r.t. $s(u)$, and all the relevant edges w.r.t. cubes in $N(s(u))$. We refer to the total overlapping length of $s(u)$ with relevant edges w.r.t. $s(u)$ as $\rel{s(u)}$. The above observation allows us to upper bound the graph diameter of any connected component associated with $u \in Q$ by $\rel{s(u)} + \sum_{s_n\in N(s(u))}\rel{s_n}$. 

It thus remains to compute, for all $u \in Q$, $\rel{s(u)}$ and $\rel{s_n}$ for all $s_n \in N(s(u))$. Let $\ell(e)$ denote the length of an edge $e \in E$. For all $e\in E$ we compute the canonical cubes of size $2^{-i}$ such that $2^{-i} = \lceil \ell(e) / 2 \rceil$. 
Denote this set of canonical cubes as $S$. For each cube $s \in S$, store the length of the overlap of the edge with the cube. Note that this correctly gives us the value $\rel{s}$ for all $s \in S$. One can now use Lemma 2.11 from~\cite{book-geometric-apprx-algo-Peled2011} to compute a compressed quadtree $Q_S$ from $S$. For any cell $v$ in $Q_S$ whose corresponding cube is not in $S$, initialize its value $\rel{v} = 0$. Next, one can use a bottom-up approach to compute, for every cell in $Q_S$ its relative edge length. For each parent, simply add the relative edge lengths stored in the children to its current value. In order to perform efficient searching on $Q_S$ we preprocess it into a finger tree, $T_f$, using similar techniques to Theorem 2.14 in~\cite{book-geometric-apprx-algo-Peled2011}.
A brief overview of the process is as follows. We construct a finger tree $T_f$ from $Q_S$ with height $O(\log cn)$ by adding a separator of $Q_S$ as the root of $T_f$ and then recursively adding the separator of the two subtrees to be the children of the previous separator.
We are now ready to compute, for each cell $u \in Q$, an upper bound on the graph diameter of any connected component. For each $u\in Q$, search for the largest cell in $T_f$ whose corresponding cube is fully contained within $s(u)$, and return the relative edge length stored at the cell. If such a cell does not exist, return zero. Repeat the search for each $s \in N(s(u))$ and add up the results. The resulting value is our upper bound on the graph diameter of any connected component of the cell $u$. We call this value $\gdub{u}$. 

When constructing the $c$-CT from $Q$ using the algorithm in Section~\ref{subsec:c_CT}, we upper bound the graph diameter of any connected component of $u \in Q$ by $\gdub{u}$.

\begin{restatable}{lemma}{correctnessDiamUb}
    \label{lem:correctnessDiamUb}
    The value $\gdub{v}$ for any cell in $v \in T$ upper bounds the graph diameter of the connected subgraph represented in the cell $v$.
\end{restatable}

\begin{proof}
    Let $u\in Q$ be the original cell in the compressed quadtree $Q$ that $v \in T$ was created from in Section~\ref{subsec:c_CT}. To argue $\gdub{v}$ is a valid upper bound on the graph diameter of the connected component corresponding to cell $v \in T$ we argue that it correctly upper bounds the graph diameter of \emph{any} connected component contained in $u$. 
    Since the original connected component in $G$ contained in any cell $u \in Q$ is fully contained within a cube of twice the radius of $u$, centered at $u$, we conclude that the graph diameter of any connected component represented by $u$ is at most $\rel{s(u)} + \sum_{s \in N(s(u)}\rel{s}$. It thus remains to argue that $\rel{s}$ for any canonical cube $s$ is correctly computed. 

    Recall that $Q_S$ denotes the compressed quadtree constructed from the set of canonical cubes $S$.
    We first argue that for all cells $w$ in $Q_S$, the value stored at the cell is equal to $\rel{s(w)}$. This follows from the fact that only edges of length at most twice the diameter contribute to this value. Since the canonical cells are the lowest cells in the compressed quadtree that overlap with relevant edges, it follows that the bottom-up summation of the values correctly maintains $\rel{s(w)}$ for all $w \in Q_S$.

    Next, for all canonical cubes $s_c$ that have non-zero relative edge length, let $w$ be a node in $Q_S$ who's corresponding cube is the cube of largest diameter fully contained inside $s_c$. We argue $\rel{s(w)} = \rel{s_c}$. If $\diam{s(w)} = \diam{s_c}$ the statement trivially follows by construction. Otherwise, if $\diam{s(w)} < \diam{s_c}$ it follows from the structure of the compressed quadtree that the parent node of $w$ can only have one child in $Q_S$, and it's length of overlap with relative edges must be the same as that of $s_c$. Finally, if $w$ does not exist in $Q_S$, it follow by construction of $Q_S$ that $\rel{s(w)} = \rel{s_c} = 0$.
\end{proof}

We are now ready to describe the construction of a WSPD$_G$ for $c$-packed graphs.

\subsubsection{Constructing the WSPD$_G$}
Construct a compressed quadtree $Q$ and compute, for each cell $u \in Q$, $\gdub{u}$ using the techniques described in Section~\ref{subsec:upperbounddiam}. Next, construct a $c$-CT, $T$, from $Q$ using the algorithm described in Section~\ref{subsec:c_CT} and store, for each cell $v \in T$, created from $u \in Q$, $\gdub{u}$ as an upper bound on the graph diameter of the connected component represented in $v$. Apply the algorithm of~\cite{book-geometric-apprx-algo-Peled2011} (Section 3.1.1) to compute a WSPD$_G$ from $T$. To determine whether a pair of nodes is well-separated we use the upper bound on the graph diameter and the Euclidean distance between their representatives.

To bound the size of the resulting WSPD$_G$ we first bound the size of $T$ by observing that each cell $u\in Q$ contains at most $O(c)$ $c$-connected sets w.r.t. $u$. 

\begin{restatable}{lemma}{repcconnected}
    \label{lem:rep_c_connect}
    For all $u \in U_{i_j}$, $1\leq j \leq \ell$, each vertex in $H_{i_j}$ within $s(u)$ represents a $c$-connected component in $H_{i_{j-1}}$ with respect to $s(u)$.
\end{restatable}

\begin{proof}
    We argue the algorithm maintains this property by induction on the number of iterations. The algorithm initially uses as input the $c$-packed graph $H_{i_1} = H_{i_0} = G$. After the first iteration of the algorithm, for each node $u \in U_{i_1}$, each surviving point $x$ in $V_{i_1}$ represents a set of points in $V_{i_0}$ that were connected to $x$ within $s^+(u)$. Since $H_{i_0} = G$ is $c$-packed, we can conclude that the distance between $x$ and any point merged into $x$ is at most $c$ times the radius of $s^+(u)$, and therefore $c$ times the diameter of $s(u)$. 
    
    Assume that the algorithm maintains the invariant for iterations 1 to $k < \ell$. Note that the graph $H_{i_k}$, used as input for iteration $k+1$, is no longer guaranteed to satisfy $c$-packedness. However, observe that after each merge step, the representative nodes in $V_{i_{k+1}}$ always correspond to original nodes present in $G$. Since the process of merging each connected set into a single node can only reduce the length of the paths between the remaining representatives, the graph distances between the remaining nodes in $H_{i_k}$ can only have decreased in comparison to their original graph distance in $G$. We therefore maintain the property that, for every node $u \in U_{i_{k+1}}$, for each representative point in $V_{i_{k+1}}$ within $u$, the points in $V_{i_k}$ that it represents can be at most $c$ times the diameter of $u$ apart in $H_{i_k}$.
\end{proof}
 With the help of Lemma~\ref{lem:rep_c_connect} one can bound the size of $T$.

\begin{restatable}{lemma}{nodespercell}
    \label{lem:number_nodes_per_cell}
    For every node $u \in Q$, the number of nodes added to $T$ is at most $O(c)$.
\end{restatable}

\begin{proof}
    From Lemma~\ref{lem:rep_c_connect} we know that each node in $V_{i_l}(s(u))$ must represent a valid $c$-connected set in $H_{i_{l-1}}(s(u))$. We now observe that one can upper bound the number of $c$-connected sets in $H_{i_{l-1}}(s(u))$ by the number of $c$-connected sets in $G(s(u))$. This observation follows from the fact that for any $1 \leq j \leq \ell$, $H_{i_j}$ was constructed by merging $c$-connected sets w.r.t. $s(u)$ in $H_{i_{j-1}}$, for every $u \in U_{i_j}$. Note that the argument used in the proof of Lemma~\ref{lem:rep_c_connect} shows that this process, for every $j$, can only cause the length of the paths between surviving nodes to become smaller. Since the position of surviving nodes does not change from their original position in $G$, it follows that the number of $c$-connected components in $H_{i_j}$ within $s(u)$ can only be smaller than that in $G$ within $s(u)$. It thus remains to argue that the number of $c$-connected components within each $s(u)$, for all $u\in V$ is at most $O(c)$. To this end, note that in order for any pair of points in $V(s)$ to be disconnected within $s^+$, every path between them must contain a subpath which intersects the boundaries of $s$ and $s^+$ and therefore has length at least $\frac{r}{\sqrt{d}}$, as shown in Figure~\ref{fig:nmbr_CC}. From the $c$-packedness of $G$ we conclude there cannot be more than $\frac{2cr\sqrt{d}}{r} = 2c\sqrt{d}$ such paths in $s^+$.
    We conclude there can be at most $O(c)$ $c$-connected sets found by the algorithm.
    \end{proof}
    \begin{figure}[H]
	\begin{center}
		\includegraphics[scale=0.45]{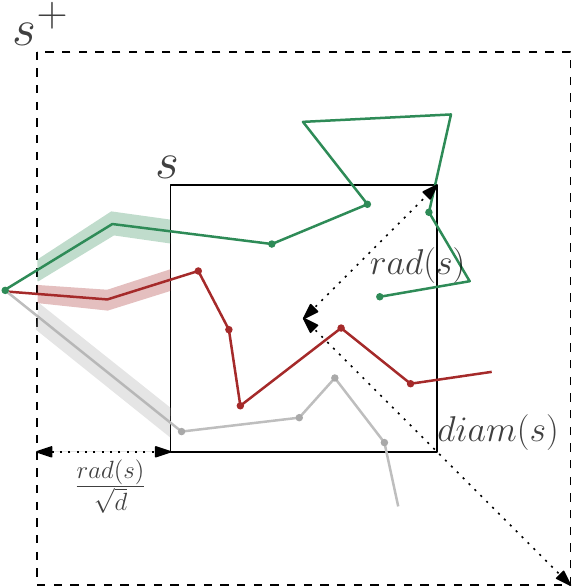}
		\caption{Two $c$-connected components in $s$ must be separated by a path of length at least $\frac{\radius{s}}{\sqrt{d}}$.}
		\label{fig:nmbr_CC}
	\end{center}
\end{figure} 

With the help of the above lemma we can analyse the size of the resulting WSPD$_G$ using a simple charging argument inspired by the one used in~\cite{book-geometric-apprx-algo-Peled2011}. 
\begin{lemma}
    \label{lem:size_and_separation_WSPDG}
    The resulting WSPD$_G$, $W$, is $\sigma$-well separated and is of size $O(c^3 \sigma n)$.
\end{lemma}
\begin{proof}
    The fact that $W$ is $\sigma$-well separated follows from our use of the Euclidean distance as a lower bound on the distance between two representatives, and our use of dub$_G$ (Section~\ref{subsec:upperbounddiam}) as upper bound on the graph diameter, to determine whether a pair is $\sigma$-well separated.
    
    To prove the size of $W$ we use a simple charging argument. For a pair $(u,v) \in W$, let $\pi(v)$ denote the parent of $v$ in $T$. Assume the last call to the algorithm in~\cite{book-geometric-apprx-algo-Peled2011} (Section 3.1.1) was via $(\pi(v), u)$. We charge the pair to $u$ and argue any node can be charged at most $O(c^2 \sigma)$ times. Since the pair $(\pi(v), u)$ was considered not well-separated by the algorithm, it follows that $d(\pi(v), u) \leq d_G(\pi(v), u) < \gdub{\pi(v)}\sigma.$
    Consider a geometric ball $B(u, r)$ of radius $r = (1+\sigma)\gdub{\pi(v)}$ centered at $u$. Note that the connected set in $G$ represented by $\pi(v)$ must be fully contained inside $B$. By $c$-packedness, there can be at most $c\sigma +1$ candidates for $\pi(v)$. Since $\pi(v)$ can have at most $2^d \cdot O(c)$ children, it follows that there are $O(c^2 \sigma)$ candidates for $v$. We conclude a node $u$ can be charged at most $O(c^2\sigma)$ in this way. By Lemma~\ref{lem:number_nodes_per_cell} there are $O(c n)$ nodes in $T$. This concludes the lemma.
\end{proof}

We now argue the runtime and space complexity of our construction.


\subsubsection{Complexity analysis}
\label{subsec:analysisWSPD}
In Lemma~\ref{lem:GSTRuntime} the runtime of the construction of the $c$-CT is analysed, as well as the size of the resulting tree. We then analyse the runtime and space complexity for computing the upper bound on the graph diameter of the subgraph represented in each cell of the $c$-CT in Lemma~\ref{lem:gdub_complexity}. Finally, with the help of these lemmas we analyse the construction time and space complexity of our WSPD$_G$ in Theorem~\ref{thm:GWSPD}. 
We start by introducing some useful terminology and a helper lemma to help bound the construction time of the $c$-CT.
\begin{definition}[New edge w.r.t $H_{i_j}$]
    We define an edge to be new w.r.t $H_{i_j}$ if the edge does not exist in $H_{i_{j-1}}$.
\end{definition}

\begin{restatable}{lemma}{chargeedges}
    \label{lem:charge_edges}
    Over all iterations $j$, the algorithm to construct a $c$-CT creates at most $O(|E|)$ new edges, where $|E|$ is the number of edges in $G$.
\end{restatable}

\begin{proof}
    We prove this lemma using a charging argument. We start allocating each edge in $H_{i_0} = G$ a credit (for a total of $|E|$ credits). We now inductively argue that in each iteration $j$ of the algorithm (i) every \emph{new} edge w.r.t. $H_{i_j}$ can be paid for using the remaining credits, and (ii) at the end of the iteration, every edge in $H_{i_j}$ can be allocated at least one unused credit. During the first iteration of the algorithm, each \emph{new} edge in $H_{i_1}$ must replace at least two edges in $E_{i_0} = E$. We use the credit from one edge to pay for the \emph{new} edge, and move the credit of the other replaced edge(s) to the new edge. Assume up to iteration $k < \ell$, the algorithm maintains invariants (i) and (ii). In iteration $k+1$, any \emph{new} edge w.r.t. $H_{i_{k+1}}$ must replace at least two edges in $E_{i_k}$. By invariant (ii) in our inductive hypothesis, we know all edges replaced by the new edge are allocated at least one unused credit. We use one credit to pay for the new edge, and allocate the remainder of the credits from the replaced edges to the new edge. This completes the proof of the lemma.
\end{proof}
We are now ready to analyse the runtime of the algorithm described in Subsection~\ref{subsec:c_CT}.
To aid the analysis, we denote the uncompressed version of $Q$ as $Q^\prime$ and denote its height as $h$. Let $U^\prime_{i_j}$, $1 \leq j \leq \ell$, be the set of nodes at level $i_j$ in the uncompressed quadtree (note that for all $1 \leq j \leq \ell$, $U_{i_j} \subseteq U^\prime_{i_j}$). We introduce some useful terminology for the edges and points contained in the cells of $Q^\prime$ for all value levels $i_j$, $1 \leq j\leq \ell$. 

For all $u \in U^\prime_{i_j}$,  $1\leq j \leq \ell$, we partition the edges in $H_{i_j}$ that have at least one endpoint in $s(u)$ into three types:

\begin{enumerate}
    \item An edge in $H_{i_j}$ is \emph{internal} with respect to $s(u)$ if both endpoints of the edge are contained inside the cell.
    \item An edge in $H_{i_j}$ is \emph{short external} if one endpoint of the edge is located inside $s(u)$, and the other is located in a neighboring cell in $U^\prime_{i_j}$.
    \item An edge in $H_{i_j}$ is \emph{long external} if one endpoint of the edge is located inside $s(u)$ and the other endpoint is located outside the cell $s(u)$ and any neighboring cell in $U^\prime_{i_j}$.
\end{enumerate}

We are now ready to analyse the construction complexity of the $c$-CT.

\begin{restatable}{lemma}{GSTRuntime}
    \label{lem:GSTRuntime}
    Given a $c$-packed graph $G$, one can construct a \GST{$c$} of size $O(cn)$ in $O(n \log n + c^2n)$ time.
\end{restatable} 

\begin{proof}
    From Lemma~\ref{lem:number_nodes_per_cell} we deduce that the size of $T$ is at most $O(cn)$. It remains to prove that $T$ can be constructed in $O(n \log n + c^2 n)$ time. Constructing the compressed quadtree $Q$ from $V$ takes $O(n \log n)$ time~\cite{book-geometric-apprx-algo-Peled2011}. Next we prove the algorithm processes at most $O(c^2 n)$ edges. 
    We prove this by considering each type of edge processed by the algorithm separately. 
    First, for all $u \in U_{i_j}$, an internal edge of $H_{i_j}$ in $s(u)$ will be processed exactly once, before being compressed into a \emph{new} edge in $H_{i_{j+1}}$ by the algorithm. From Lemma~\ref{lem:charge_edges} we deduce there can only be a linear number of internal edges processed throughout the execution of the algorithm. 
    Next, we argue the number of occurrences of short external edges. Since, by the end of iteration $j-1$, there can be at most $O(c)$ points of $H_{i_j}$ in $V(v)$ for any $v \in U_{i_{j-1}}$, there can be at most $2^d O(c)$ points of $H_{i_j}$ in $V(u)$ for any $u \in U_{i_j}$\footnote{This follows from the fact that a cell in $Q$ can have at most $2^d$ non-empty children.}. Observe that for each short external edge in $u \in U_{i_j}$, there are at most $(3^d - 1)$ neighboring cells in $U^\prime_{i_j}$ in which the other endpoint can lie. By Lemma~\ref{lem:number_nodes_per_cell}, $u$ can therefore only share $O(c^2)$
    edges with each neighbor. We conclude there can be at most $O(c^2)$
    short external edges for each cell in $Q$. 
    Finally, we observe that each long edge that is \emph{new} in $H_{i_j}$ with respect to $s(u)$, must replace at least one edge that was already long in $G$ with respect to $s(u)$\footnote{This follows from the fact that a representative point for a connected set is always chosen from the same cell as all the members of the connected set.}. Since the length of a long external edge in $G$ with respect to $s(u)$ must have length at least $\diam{u}$, it follows from the $c$-packedness of $G$ that there can be at most $O(c)$ long edges in $G$ with at least one endpoint in $s(u)$. We conclude that there can be at most $O(c)$ such edges for each node $u \in U_{i_j}$ across all value levels of $Q$. 
\end{proof}

Next, we analyse the runtime and space complexity associated with computing the upper bound on the graph diameter of any subgraph represented in each cell of the $c$-CT.

\begin{restatable}{lemma}{gdubcomplexity}
    \label{lem:gdub_complexity}
    Let $G$ be a $c$-packed graph and $Q$ its corresponding compressed quadtree. One can compute an upper bound on the graph diameter of any $c$-connected component contained in $u$, for all $u\in Q$, in $O(cn \log n)$ time, using at most $O(cn)$ space.
\end{restatable}
\begin{proof}
    Computing the set of canonical cells from the edges of $G$ takes $O(|E|) = O(cn)$ time\footnote{This follows from the fact that each edge $e$ can overlap with at most $O(d)$ canonical cells of diameter at least $\lceil \ell(e) / 2 \rceil$.}. The number of canonical cells in the set $S$ is therefore at most $O(cn)$. Constructing a compressed quadtree from these cells can be done in $O(cn \log cn)$ time and $O(cn)$ space using the techniques described in Section 2.2.3 in~\cite{book-geometric-apprx-algo-Peled2011}. Computing the relative edge length for each cell in $Q^\prime$ using the bottom-up approach can now be done in $O(cn)$.
    Constructing a finger tree $T_f$ for efficient searching can be done in $O(cn)$ time, using $O(cn)$ space. Note that the height of $T_f$ is now $O(\log cn)$. Therefore, finding the node in $T_f$ corresponding to a cube $s(u)$ of a cell $u\in Q$ and every cube $s \in N(s(u))$ takes at most $3^d \cdot O(\log cn)$ time. Therefore, we can compute $\gdub{u}$ for all $u \in Q$ in $O(cn \log n)$ time using $O(cn)$ space. 
\end{proof}
With the help of the above lemmas we obtain the following theorem.
\GWSPDthm*
\begin{proof}
    Constructing the compressed quadtree $Q$ requires $O(n \log n)$ time and $O(n)$ space. Upper bounding the diameter using the techniques introduced in Section~\ref{subsec:upperbounddiam} can be done in $O(cn \log n)$ time, using $O(cn)$ space (Lemma~\ref{lem:gdub_complexity}). We conclude from Lemma~\ref{lem:GSTRuntime} that we can construct a $c$-CT, $T$, from $Q$ in $O(c^2n)$ time, using $O(cn)$ space. Finally, using the construction from~\cite{book-geometric-apprx-algo-Peled2011} we can compute the WSPD$_G$ from $T$ in $O(c^3\sigma n)$ time using $O(cn)$ space. This concludes the proof. 
\end{proof}

\section{A Separator Theorem for \texorpdfstring{$c$}{c}-packed Graphs}\label{sec:SmallSeparator}

A subset of vertices $C$ of a graph with $n$ vertices is called a (balanced) separator if the remaining vertices can be partitioned into two sets $A$ and $B$ such that there are no edges between $A$ and $B$, with $|A|< \alpha\cdot n$ and $|B|<\alpha \cdot n$ for some constant $1/2\leq \alpha <1$. The subset $C$ is said to \emph{$\alpha$-separate} the graph.

In this section we show how to compute a separator of size $O(c)$ for any $c$-packed graph that $(1-\frac{1}{2\lambda^3})$-separates the graph. The constant $\lambda$ is the doubling constant of $\R^d$.

We have the following separating lemma.
\begin{lemma}~\label{lem:separator-ring-deterministic}
     Given a set $P$ of $n$ points in $\R^d$, where $d$ is fixed, one can compute a ball $b(p,r)$, which is centered at $p$ and has radius $r$, such that $\mathbf{b}(p,r)$ contains at least $n/(2\lambda^3)$ points of $P$, where $\lambda$ is the doubling constant of $\R^d$, and $\mathbf{b}(p,2r)$ of twice the radius contains at most $n/2$ points of $P$. The running time of the algorithm is $O(\lambda^{3d} n)$.
\end{lemma}
\begin{proof}
    Let $r_{opt}(P,k)$ be the radius of the smallest ball enclosing $k$ points of $P$. Har-Peled and Mazumdar~\cite{book-geometric-apprx-algo-Peled2011} gave an algorithm for computing a $k$-enclosing ball with radius at most $2r_{opt}(P,k)$ in $O(n(n/k)^d)$ time. 
    As a result, one can compute a $(n/(2\lambda^3))$-enclosing ball $\mathbf{b}(p,r)$ in $O(\lambda^{3d}\cdot n)$ time.

    Now one can prove that $\mathbf{b}(p,2r)$ contains at most $n/2$ points. Since $r\leq 2r_{opt}(P,n/(2\lambda^3))$, by the doubling property, $\mathbf{b}(p,2r)$ can be covered by at most $\lambda^3$ balls of radius $r_{opt}(P,n/(2\lambda^3))/2$. Note that any ball of radius $r_{opt}(P,n/(2\lambda^3))/2$ contains strictly less than $n/(2\lambda^3)$ points. Thus $\mathbf{b}(p,2r)$ contains at most $n/2$ points. This concludes the proof of the lemma. 
\end{proof}

Based on Lemma~\ref{lem:separator-ring-deterministic}, we can show that one can construct separators of size $O(c)$ that $(1-\frac{1}{2\lambda^3})$-separate $c$-packed graphs, as shown in Theorem~\ref{thm:separator-technical-overview}.

\separator*

\begin{proof}
    First we construct a separator $C$. Then we prove that $C$ satisfies the required properties. Finally we analyze the running time of the algorithm.

    Let $P$ be the vertices of $c$-packed graph $G$ and run the algorithm described in Lemma~\ref{lem:separator-ring-deterministic}. Let $\mathbf{b}(p,r)$ be the computed ball that contains at least $n/(2\lambda^3)$ vertices of $G$. Let $\hat{A}$ be the vertices of $G$ contained in $\mathbf{b}(p,r)$ and let $\hat{B}$ be the vertices of $G$ lying outside $\mathbf{b}(p,2r)$. See Figure~\ref{fig:thm-small-separator} for an illustration. 

    Now set up a max-flow instance. A similar construction was used in Lemma~11 of~\cite{DBLP:conf/soda/GudmundssonSW23}. Set the capacity of each edge of $G$ to 1. Let the vertices in $\hat{A}$ be the sources, and let the vertices in $\hat{B}$ be the sinks. Run the Ford-Fulkerson algorithm on the instance. According to the max-flow min-cut theorem, the max-flow of the instance is equal to its min-cut. Let $(S,T)$ be a min-cut and let $\{e_1,\ldots,e_l\}$ be the set of edges from $S$ to $T$, where $\hat{A}\subseteq S$ and let $\hat{B}\subseteq T$. Choose one endvertex for each edge in $\{e_1,\ldots,e_l\}$ and add it to $C$ ($C$ is initially empty). Let $A=S\setminus C$ and let $B=T\setminus C$. This finishes the construction of $C$.

    \begin{figure}[htb]
    \centering
    \includegraphics[width=.35\textwidth]{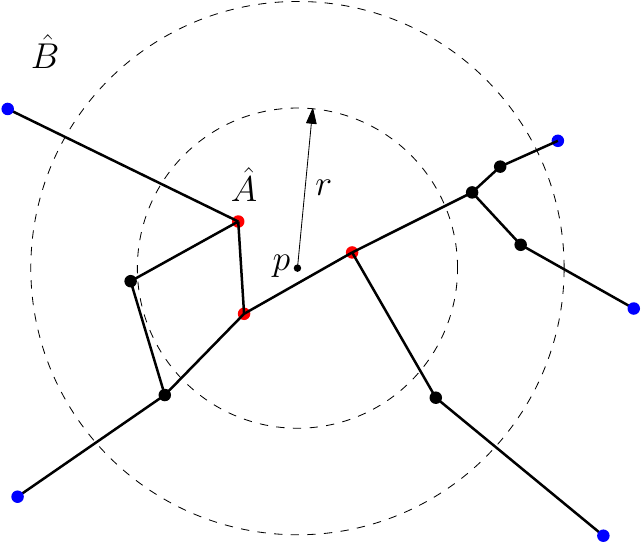}
    \caption{Illustration of the proof of Theorem~\ref{thm:separator-technical-overview} where $p$ is the center of ball $\mathbf{b}(p,r)$. Vertices in $\hat{A}$ are drawn in red. Vertices in $\hat{B}$ are drawn in blue. The value of the min-cut in the figure is 4.}
    \label{fig:thm-small-separator}
    \end{figure}

    Next we show that $C$ (i) has size $O(c)$, and (ii) $(1-\frac{1}{2\lambda^3})$-separates $G$. Property (ii) follows from $(S,T)$ being a min-cut. Since $\hat{A}$ contains at least $n/(2\lambda^3)$ vertices of $G$ and $\hat{A}\subset S$, $A$ contains at least $n/(2\lambda^3)$ vertices of $G$. Since $\hat{B}$ contains at least $n/2$ vertices of $G$ and $\hat{B}\subset T$, $B$ contains at least $n/2$ vertices of $G$. Property (i) follows from $c$-packedness. In the max-flow instance, the value of the max-flow is $l$. Since all the edges have capacity 1, there are $l$ edge-disjoint paths from $\hat{A}$ to $\hat{B}$. Each such path pierces both the inner and the outer boundaries of the $d$-dimensional spherical shell $\mathbf{b}(p,2r)\setminus \mathbf{b}(p,r)$. Due to $c$-packedness, the length of edges inside the spherical shell is at most $2cr$. Since the width of the spherical shell is $r$, there are at most $2c$ edge-disjoint paths from $\hat{A}$ to $\hat{B}$. Therefore $|C|=l\leq 2c$.

    Finally we analyze the running time of the algorithm. Running the algorithm in Lemma~\ref{lem:separator-ring-deterministic} on the vertices of $G$ takes $O(\lambda^{3d} n)$
    time. We use the Ford-Fulkerson algorithm to solve the max-flow instance. The running time of the algorithm is proportional to the number of edges in $G$ times the value of the max-flow. Since $l\leq 2c$ and there are $O(cn)$ edges in $G$, solving the max-flow instance takes $O(c^2n)$ time. Since $\lambda=O(d)$, the overall running time of the algorithm is $O(c^2n)$.
\end{proof}

In the next Section we show how one can use the above results to obtain an exact distance oracle for $c$-packed graphs. 

\section{An Exact Distance Oracle for $c$-packed Graphs}
\label{sec:Dist_oracles}
We can combine our separator in Theorem~\ref{thm:separator-technical-overview} with standard techniques~\cite{ShortestPathAmongObstaclesPlane_ArikatiCCDSZ96} to construct an exact distance oracle via a separator hierarchy~\cite{Planar-separator-theorem_LiptonTarjan79}. 

We apply Theorem~\ref{thm:separator-technical-overview} recursively to decompose the given $c$-packed graph $G$. Applying Theorem~\ref{thm:separator-technical-overview} on $G$ partitions the vertices of $G$ into subsets $A$, $B$ and $C$ where $C$ is a separator of size $O(c)$. Let $G(A)$ be the subgraph of $G$ induced by vertices in $A$ and let $G(B)$ be the subgraph of $G$ induced by the vertices in $B$. Recursively apply the separator theorem on $G(A)$ and $G(B)$ until the subgraphs have constant size, and we get a balanced decomposition tree $\mathcal{T}$. Each node in $\mathcal{T}$ corresponds to a subgraph of $G$ and is associated with a separator of the subgraph. The root of $\mathcal{T}$ corresponds to $G$ and is associated with the separator $C$, and so on.  For each node $v$ in $\mathcal{T}$, compute shortest path trees rooted at each vertex in the separator associated with $v$. Note that the separators associated with the internal nodes of $\mathcal{T}$ and the subgraphs of the leaves of $\mathcal{T}$ form a partition of the vertices of $G$. Map each vertex in $G$ to the node in $\mathcal{T}$ whose separator (internal node) or subgraph (leaf node) contains the vertex. Finally preprocess $\mathcal{T}$ in linear time so that LCA (Lowest Common Ancestor) queries can be answered in constant time~\cite{LCArevisited_BenderColton2000}. This finishes the preprocessing of the distance oracle.

To answer a distance query between any two vertices $u$ and $v$ in $G$, use the mapping for vertices in $G$ and the LCA structure to locate the node $s$ in $\mathcal{T}$, whose associated separator separates $u$ from $v$ in the subgraph of $s$. Let $G_s$ denote the subgraph of $s$ and let $sep(s)$ denote the separator associated with $s$. Use the shortest path trees stored at $s$ to compute $d_{G_s}(u,v)=\min_{t\in sep(s)}\{d_{G_s}(u,t)+d_{G_s}(t,v)\}$. Let $anc(s)$ be the ancestors of $s$ in $\mathcal{T}$. For each node $a$ in $anc(s)$, compute $d_{G_a}(u,v)=\min_{t\in sep(a)}\{d_{G_a}(u,t)+d_{G_a}(t,v)\}$. Return $\min_{a\in anc(s)}\{d_{G_s}(u,v),d_{G_a}(u,v)\}$ as the answer.

We can prove the correctness of the query algorithm inductively. Starting from the root $o$ of $\mathcal{T}$, if the shortest path from $u$ to $v$ goes through the separator associated with $o$, then the distance from $u$ to $v$ is $d_{G_o}(u,v)$. Otherwise the shortest path from $u$ to $v$ does not go through the separator associated with $o$, therefore must stay in the subgraph of $y$, where $y$ is $o$'s child that is an ancestor of $s$. If the shortest path from $u$ to $v$ goes through the separator associated with $y$, the distance from $u$ to $v$ is $d_{G_y}(u,v)$. And so on until we reach $s$. This proves that $\min_{a\in anc(s)}\{d_{G_s}(u,v),d_{G_a}(u,v)\}$ is the distance from $u$ to $v$, therefore the correctness of the query algorithm.

Finally we analyze the complexities of the preprocessing algorithm and the query algorithm. Since $\mathcal{T}$ is a balanced binary tree, its height is $O(\log n)$. For each level of $\mathcal{T}$, computing the shortest path trees rooted at the separators takes $O(c^2n + cn\log n)$ time, using Dijkstra's algorithm. Computing the separators at each level takes $O(c^2n)$
time according to Theorem~\ref{thm:separator-technical-overview}. There are $O(\log n)$ levels, thus the total preprocessing time is $O(c^2n\log n+cn\log^2 n)$
. Since there are $O(c)$ vertices in an associated separator, computing $d_{G_s}(u,v)$ or $d_{G_a}(u,v)$ takes $O(c)$ time. Therefore the query time is $O(c\log n)$. The shortest path trees stored at each level of $\mathcal{T}$ takes $O(cn)$ space, and the total space of preprocessing is $O(cn\log n)$. We have obtained the following theorem.

\exactdistoracle*

To build the tree cover in the next section, we need a WSPD$_G$ that uses graph distances between representatives instead of Euclidean distances between representatives. Let $a$ and $b$ be two nodes in the $c$-CT. The pair $\{C_a,C_b\}$ is well-separated w.r.t graph distance if $\max\{diam_G(C_a),diam_G(C_b)\}\cdot \sigma\leq dist_G(rep_a,rep_b)-diam_G(C_a)-diam_G(C_b)$. From Theorem~\ref{thm:GWSPD} in Section~\ref{sec:GWSPD} and Theorem~\ref{thm:exact-distance-oracle}, we have
\begin{corollary}
    \label{thm:GWSPD-graph-distance}
    Given a $c$-packed graph $G$ in $\R^d$, one can construct a WSPD$_G$ with separation factor $\sigma$, where all distances are graph distances, in $O(c^4\sigma n\log n+cn\log ^2 n)$ time, using $O(cn\log n)$ space. The WSPD$_G$ has size $O(c^3\sigma n)$.
\end{corollary} 

\section{A tree cover for \texorpdfstring{$c$}{c}-packed graphs}\label{sec:tree_cover}
In this section, we show a $(1+\varepsilon)$-distortion tree cover for $c$-packed metrics with constant number of trees when $c$ is $O(1)$.

We use the concept of ``dumbbells'', which was used in the ``Dumbbell Theorem" for Euclidean metrics by Arya~\etal~\cite{dumbbell-theorem-Euclidean-stoc95}. We construct the dumbbells using the WSPD$_G$ in Corollary~\ref{thm:GWSPD-graph-distance}. Recall that all the distances in the WSPD$_G$ are graph distances. We then partition the linear number of dumbbells into dumbbell groups each of which satisfy (i) the \emph{length-grouping property} and (ii) the \emph{empty-region property}. Finally we connect the dumbbells in each group hierarchically and build a ``dumbbell tree''. 

The rest of this section is organized as follows. We describe the dumbbells in Section~\ref{ssec:dumbbells}, prove the packing lemmas in Section~\ref{ssec:packing-lemmas-tc}, and show how to partition the dumbbells into groups satisfying the length-grouping property and the empty-region property in Section~\ref{ssec:partition-dumbbells}. We construct the dumbbell trees in Section~\ref{ssec:build-dumbbell-trees} and prove that the dumbbell trees constitutes a tree cover in Section~\ref{ssec:dumbbell-trees-constitute-tc}.

\subsection{Dumbbells}\label{ssec:dumbbells}
Throughout Section~\ref{sec:tree_cover}, let $T$ denote $T$ and let $\{A_i,B_i\}$, $1\leq i\leq m$, be the well-separated pairs generated by the WSPD$_G$ algorithm in Corollary~\ref{thm:GWSPD-graph-distance}. For any node $v$ in $T$, let $C_v$ denote the connected set of $v$ and let $rep_v$ denote the representative point of $C_v$. $A_i$ or $B_i$ is the connected set of some node in $T$. Let $A_i=C_a$ and let $B_i=C_b$, where $a,b$ are nodes in $T$. 

For each pair $\{A_i,B_i\}$, $1\leq i\leq m$, we construct a dumbbell $D(a,b)$. The dumbbell $D(a,b)$ has two \emph{heads}, $head(a)$ and $head(b)$. $head(a)$ ($head(b)$) is the graph metric ball that is centered at $rep_a$ ($rep_b$), has radius $\gdub{C_a}$ ($\gdub{C_b}$) and \emph{only} includes the points in $C_a$ ($C_b$). Recall that $\gdub{C_a}$ is the upper bound of the graph diameter of $C_a$. We refer to $C_a$ as the \emph{head set} of $head(a)$. $head(a)$ and $head(b)$ are joined by a graph metric segment of length $dist_G(rep_a,rep_b)$. See Figure~\ref{fig:dumbbells} for an illustration.
    
    \begin{figure}[H]
	\centering
		\includegraphics[width=.8\textwidth]{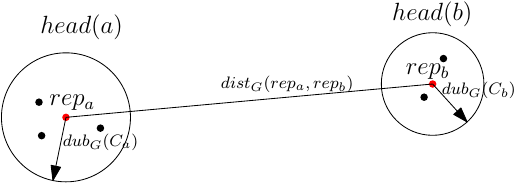}
		\caption{The dumbbell $D(a,b)$ for the pair $\{A_i=C_a,B_i=C_b\}$. $head(a)$ ($head(b)$) is the ball that is centered at $rep_a$ ($rep_b$), has radius $\gdub{C_a}$ ($\gdub{C_b}$) and only includes points in $C_a$ ($C_b$). }
		\label{fig:dumbbells}
    \end{figure}

\subsection{Packing lemmas}\label{ssec:packing-lemmas-tc}

In this section we prove a few packing lemmas, which are required for establishing the empty-region property and constructing the dumbbell trees. The first lemma bounds the number of disjoint heads whose parent heads have size proportional to the hypercube $\mathbf{s}$. The proof is similar to Lemma~9.4.3 in~\cite{book-geom-spanner-networks}.
\begin{lemma}\label{lem:hypercube-intersecting}
    Let $\mathbf{s}$ be a hypercube in $\R^d$, let $\ell$ be the side length of $\mathbf{s}$ and let $\alpha$ be a positive real number. Let $b_1,b_2,\ldots,b_k$ be nodes of $T$ such that
    \begin{enumerate}
        \item $b_i$ is not the root of $c$-CT for all $1\leq i\leq k$,
        \item the sets $C_{b_i}$, $1\leq i\leq k$, are pairwise disjoint,
        \item $L(s(\pi(b_i)))\geq \ell/\alpha$ for $1\leq i\leq k$, and
        \item $s(b_i)\cap \mathbf{s}\neq \emptyset$ for all $1\leq i\leq k$.
    \end{enumerate}
    Then $k\leq c(2\alpha +2)^d$.
\end{lemma}

The following lemma states that the radius of a head of a dumbbell is at most $\frac{1}{\sigma+1}$ of the length of the dumbbell. Recall that $\gdub{C_a}$ is the radius of $head(a)$.
\begin{lemma}\label{lem:head-size-to-length}
    Let $D(a,b)$ be a dumbbell and let $\ell(a,b)$ be its length. Let $p\in C_a$ and let $q\in C_b$. Then $\gdub{C_a}\leq \frac{\ell(a,b)}{\sigma+1}$ and $\gdub{C_b}\leq\frac{\ell(a,b)}{\sigma+1}$. Also $\gdub{C_a}\leq \frac{dist_G(p,q)}{\sigma}$ and $\gdub{C_b}\leq \frac{dist_G(p,q)}{\sigma}$.
\end{lemma}
\begin{proof}
    Since $C_a$ and $C_b$ are well-separated with respect to $\sigma$, we have 
    \begin{align*}
        \max\{\gdub{C_a},\gdub{C_b}\}\cdot \sigma\leq \ell(a,b)-\gdub{C_a}-\gdub{C_b}\leq dist_G(p,q).
    \end{align*}
    Thus $\gdub{C_a}\leq\frac{\ell(a,b)}{\sigma+1}$ and $\gdub{C_b} \leq\frac{\ell(a,b)}{\sigma+1}$. Also $\gdub{C_a}\leq \frac{dist_G(p,q)}{\sigma}$ and $\gdub{C_b}\leq \frac{dist_G(p,q)}{\sigma}$.
\end{proof}

The heads $head(a)$, $head(b)$ of a dumbbell $D(a,b)$ might be very small compared to the length of $D(a,b)$. Let $L(s(a))$ denote the side length of cube $s(a)$. The following lemma states that the heads $head(\pi(a))$, $head(\pi(b))$ are at least proportional to the length of $D(a,b)$. 
\begin{lemma}\label{lem:head-size}
    Let $D(a,b)$ be a dumbbell and let $\ell$ be its length. Then $\gdub{C_{\pi(a)}}\geq\frac{\ell}{\sigma + 4}$ and $\gdub{C_{\pi(b)}}\geq\frac{\ell}{\sigma + 4}$. Also, $L(s(\pi(a)))\geq \frac{\ell}{c\sqrt{d}(\sigma+4)}$ and $L(s(\pi(b)))\geq \frac{\ell}{c\sqrt{d}(\sigma+4)}$.
\end{lemma}
\begin{proof}
    Without loss of generality, assume that $\gdub{C_{\pi(a)}}\leq \gdub{C_{\pi(b)}}$, thus it suffices to prove that $\gdub{C_{\pi(a)}}\geq \frac{\ell}{\sigma+4}$. Let $u$ be the lowest common ancestor of $a$ and $b$ in $T$, let $v$ be the child of $u$ that is ancestor of $a$ and let $w$ be the child of $u$ that is ancestor of $b$. 

    If $a=v$, then $\pi(a)=u$. Since $C_a\subseteq C_u$ and $C_b\subseteq C_u$, the length of $D(a,b)$ is at most the diameter of $C_u$. Thus $\ell\leq \gdub{C_{\pi(a)}}$, and the lemma holds in this case.

    If $a\neq v$, let $\pi^0(a)=a$ and let $\pi^i(a)=\pi(\pi^{i-1}(a))$ for $i\geq 1$. When computing the WSPD$_G$ in Corollary~\ref{thm:GWSPD-graph-distance}, we find well-separated pairs for the nodes pair $(v,w)$, which include the pair $\{C_a,C_b\}$. Let $k\geq 0$ be the integer such that finding well-separated pairs for $(\pi(a),\pi^k(b))$ recursively finds well-separated pairs for $(a,\pi^k(b))$ (note that such $k$ exists). Thus $C_{\pi(a)}$ and $C_{\pi^k(b)}$ are not well-separated, and we have $\gdub{C_{\pi^k(b)}} \leq \gdub{C_{\pi(a)}}$. 

    Assume that $\gdub{C_{\pi(a)}}<\frac{\ell}{\sigma + 4}$, we will show that the assumption leads to a contradiction. Let the representatives (centers) of $C_{a}$ and $C_{b}$ be $x$ and $y$, respectively. We have $\ell=dist_G(x,y)$. Let the representatives of $C_{\pi(a)}$ and $C_{\pi^k(b)}$ be $x'$ and $y'$, respectively. Since $x,x'\in C_{\pi(a)}$, 
    \begin{align*}
        dist_G(x,x')\leq \gdub{C_{\pi(a)}}.
    \end{align*} 
    Since $y,y'\in C_{\pi^k(b)}$, 
    \begin{align*}
        dist_G(y,y')\leq \gdub{C_{\pi^k(b)}}\leq \gdub{C_{\pi(a)}}.
    \end{align*} 
    By the triangle inequality, $\ell=dist_G(x,y)\leq dist_G(x,x')+dist_G(x',y')+dist_G(y',y)$. Thus 
    \begin{align*}
        dist_G(x',y')\geq \ell-2\cdot \gdub{C_{\pi(a)}} > \ell - \frac{2}{\sigma+4}\ell.
    \end{align*}

    Let $B_1$ ($B_2$) be the metric ball which is centered at $x'$ ($y'$) and has radius $\frac{\ell}{\sigma+4}$. The graph distance between $B_1$ and $B_2$ is at least 
    \begin{align*}
        dist_G(x',y')-\frac{2}{\sigma+4}\ell > \sigma \cdot \frac{\ell}{\sigma+4}.
    \end{align*}
    This implies that $B_1$ and $B_2$ are well-separated. By the assumption, $B_1$ contains $C_{\pi(a)}$ and $B_2$ contains $C_{\pi^k(b)}$. Therefore $C_{\pi(a)}$ and $C_{\pi^k(b)}$ are well-separated, which is a contradiction. 

    According to the definition of $\gdub{\cdot}$ and the $c$-packedness property, we have 
    \begin{align*}
        3^d\cdot c\cdot \sqrt{d}L(s(\pi(a)))\geq \gdub{\pi(a)}\geq \frac{\ell}{\sigma+4}.
    \end{align*}
    Thus $L(s(\pi(a)))\geq \frac{\ell}{3^dc\sqrt{d}(\sigma+4)}$. $L(s(\pi(b)))\geq \frac{\ell}{3^dc\sqrt{d}(\sigma+4)}$ follows from the same argument.
\end{proof}

The following lemma shows that the size of a head decreases by a factor of 2 after descending a constant number of levels.
\begin{lemma}\label{lem:descend-head-size}
    Let $b$ and $b'$ be two nodes in $T$ where $b'$ is an ancestor of $b$. If the path from $b'$ to $b$ contains at least $1+d+\log (4c\sqrt{d})$ edges, then $\gdub{C_{b}}\leq \frac{\gdub{C_{b'}}}{2}$.
\end{lemma}
\begin{proof} 
    We first prove a lower bound on $\gdub{C_{b'}}$. Let $b_0'=b',b_1',\ldots,b_h'=b$ be the path from $b'$ to $b$ in $T$. When descending from $b_i'$ to $b_{i+1}'$, $0\leq i<h$, there must be a point in $C_{b'}$ lying in the half space delimited by the hyperplane through one face of $s(b_i')$ and not containing $s(b_i')$. Therefore descending from $b'$ for at most $1+d$ levels, at $b_j'$, $1\leq j\leq d+1$, either (i) there are two points in $C_{b'}$ lying in two half spaces that are (respectively) delimited by one of the hyperplanes through two opposite faces of $s(b_j')$ and do not contain $s(b_j')$, or (ii) there is one point in $C_{b'}$ that is separated from $s(b_j')$ by a slab of width at least $L(s(b_j'))$. In Figure~\ref{fig:lem-edges-between-heads}(a), $s(b_j')$ is the cell containing the red point, the two blue points lie in the two half spaces delimited by the hyperplanes through two opposite faces of $s(b_j')$. In Figure~\ref{fig:lem-edges-between-heads}(b), the blue point is separated from $s(b_j')$ by a slab of width at least $L(s(b_j'))$. In either case, we know that $\gdub{C_{b'}}\geq L(s(b_j'))$.

    On descending one level, the size of a quadtree cell decreases by a factor of at least 2. Thus $L(s(b_{j+\log (4c\sqrt{d})}'))\leq \frac{L(s(b_j'))}{4c\sqrt{d}}$. From the $c$-packedness property, we know that $\gdub{C_{b_{j+\log (4c\sqrt{d})}'}}$, is at most $\frac{L(s(b_j'))}{2}$. Thus $\gdub{C_b}\leq \frac{\gdub{C_{b'}}}{2}$. 
\end{proof}
\begin{figure}[H]
	\centering
		\includegraphics[width=.8\textwidth]{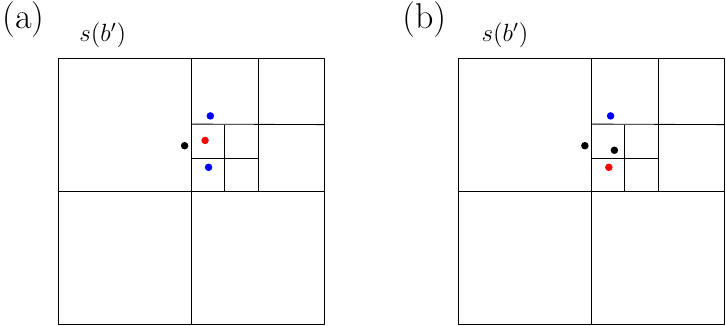}
		\caption{(a) $s(b_j')$ is the cell containing the red point. The two blue points lie in two half spaces which are (respectively) delimited by (one of) the hyperplanes through two opposite faces of $s(b_j')$ and do not contain $s(b_j')$. (b) $s(b_j')$ is the cell containing the red point. The blue point is separated from $s(b_j')$ by a slab of width at least $L(s(b_j'))$.}
		\label{fig:lem-edges-between-heads}
\end{figure}

Let $D(a,b)$ and $D(a',b')$ be two dumbbells whose length are in $[\ell,2\ell]$ and $b'$ is an ancestor of $b$. The next lemma gives a constant upper bound on the number of edges on the path from $b'$ to $b$ in $T$. 
\begin{lemma}\label{lem:edges-on-path-between-heads}
    Let $\ell > 0$, and let $D(a,b)$ and $D(a',b')$ be two dumbbells whose length are in $[\ell,2\ell]$. Let $b'$ be an ancestor of $b$ in $T$, and let $h$ be the number of edges on the path between $b'$ and $b$ in $T$. We have 
    \begin{align*}
        h\leq 2+d+\log (4c\sqrt{d}) + (1+d+\log (4c\sqrt{d}))\log \left(\frac{2(\sigma+4)}{\sigma+1}\right).
    \end{align*}
\end{lemma}
\begin{proof}
    If the number of edges between $b'$ and $b$ is at most $2+d+\log (4c\sqrt{d})$, the lemma holds trivially. Thus in the following we assume that $h>2+d+\log (4c\sqrt{d})$.

    According to Lemma~\ref{lem:head-size-to-length}, $\gdub{C_{b'}}$ is at most $\frac{1}{\sigma+1}$ times the length of $D(a',b')$. Thus 
    \begin{align*}
        \gdub{C_{b'}}\leq \frac{2\ell}{\sigma+1}.
    \end{align*}
    From Lemma~\ref{lem:head-size} we have 
    \begin{align*}
        \gdub{C_{\pi(b)}}\geq \frac{\ell}{\sigma+4}.
    \end{align*}
    Therefore $\frac{\gdub{C_{b'}}}{\gdub{C_{\pi(b)}}}\leq \frac{\frac{2\ell}{\sigma+1}} {\frac{\ell}{\sigma+4}}=\frac{2(\sigma+4)}{\sigma+1}.$
    According to Lemma~\ref{lem:edges-on-path-between-heads}, $\gdub{\cdot}$ decreases by a factor of at least 2 if descending $1+d+\log (4c\sqrt{d})$ levels. Thus 
    \begin{align*}
        2^{\lfloor\frac{h-1}{1+d+\log (4c\sqrt{d})}\rfloor}\leq \frac{2(\sigma+4)}{\sigma+1}.
    \end{align*}
    Since $\frac{h-1}{1+d+\log (4c\sqrt{d})}-1 \leq \lfloor\frac{h-1}{1+d+\log (4c\sqrt{d})}\rfloor$, the lemma hence follows.
\end{proof}

The following lemma shows an upper bound on the number of dumbbells having a fixed head and having similar length.
\begin{lemma}\label{lem:dumbbells-with-fixed-head}
    Let $a$ be a node of $T$ and let $\ell>0$. The number of dumbbells having $head(a)$ as one head and having length in $[\ell,2\ell]$ is at most $c(8c\sqrt{d}(\sigma+4)+2)^d$.
\end{lemma}
\begin{proof}
    Let $b$ be a node of $T$ such that $D(a,b)$ is a dumbbell whose length is in $[\ell,2\ell]$. Let $x$ be the center of $head(a)$ and let $y$ be the center of $head(b)$. We have $dist_G(x,y)\leq 2\ell$. 
    
    To bound the number of such nodes like $b$, consider the hypercube $\mathbf{s}$ which is centered at the center of $head(a)$ and has side length $4\ell$. We have $s(b)\cap \mathbf{s}\neq \emptyset$. According to Lemma~\ref{lem:head-size}, $L(s(\pi(b)))\geq \frac{\ell}{3^dc\sqrt{d}(\sigma+4)}.$

    Applying Lemma~\ref{lem:hypercube-intersecting} with $\ell=4\ell$ and $\alpha=43^dc\sqrt{d}(\sigma+4)$, we have the number of such nodes as $b$ is at most $c(8c3^d\sqrt{d}(\sigma+4)+2)^d$.
\end{proof}

Now we are ready to prove the main packing property of dumbbells with similar length, which is critical for establishing the empty-region property.
\begin{lemma}\label{lem:main-packing-tc}
    Let $\gamma$ and $\ell$ be positive numbers. Let $D(u,v)$ be a dumbbell with length in $[\ell,2\ell]$. The number of dumbbells $D(a,b)$ such that 
    \begin{enumerate}[(i)]
        \item the length of $D(a,b)$ is in $[\ell,2\ell]$,
        \item at least one of $head(a)$ and $head(b)$ is within graph distance $\gamma \ell$ to one head of $D(u,v)$
    \end{enumerate}
    is $t_{\sigma\gamma}=O(c^{2d+2}\sigma^d(1+\gamma\sigma)^d)$.
\end{lemma}
\begin{proof}
    Let $\overline{\mathcal{D}}$ denote the set of dumbbells $D(a,b)$ such that at least one head of $D(a,b)$ is within graph distance $\gamma\ell$ to the head $head(u)$ of $D(u,v)$. We will give an upper bound for $|\overline{\mathcal{D}}|$.

    Without loss of generality, assume that the head $head(a)$ of $D(a,b)$ is within graph distance $\gamma \ell$ to $head(u)$. From Lemma~\ref{lem:head-size-to-length}, we have $\gdub{C_a}\leq \frac{2\ell}{\sigma+4}$ and $\gdub{C_u}\leq \frac{2\ell}{\sigma+4}$. Thus the graph distance between the center of $head(a)$ ($rep_a$) and the center of $head(u)$ ($rep_u$) is at most 
    \begin{align*}
        \gamma\ell+\frac{4\ell}{\sigma+4}=(\gamma+\frac{4}{\sigma+4})\ell.
    \end{align*}
    Observe that when graph distance between $rep_a$ and $rep_u$ is at most $(\gamma+\frac{4}{\sigma+4})\ell$, $rep_a$ must lie in the hypercube $\mathbf{s}$ which is centered at $rep_u$ and has side length $(2\gamma+\frac{8}{\sigma+4})\ell$, thus $s(a)\cap \mathbf{s}\neq \emptyset$. 

     Let $a_1,a_2,\ldots,a_k$ be the nodes in $T$ such that $s(a_i)\cap \mathbf{s}\neq \emptyset$, $1\leq i\leq k$. According to Lemma~\ref{lem:dumbbells-with-fixed-head}, we have $|\overline{\mathcal{D}}|\leq c(8c3^d\sqrt{d}(\sigma+4)+2)^d\cdot k.$
     
     To upper bound $k$, one can not directly apply Lemma~\ref{lem:hypercube-intersecting} since the sets $C_{a_i}$, $1\leq i\leq k$, are not pairwise disjoint. Let $a_1',a_2',\ldots,a_{k'}'$ be the nodes in $a_1,a_2,\ldots,a_k$ such that: (i) each $a_i$, $1\leq i\leq k$, is in the subtree rooted at $a_j'$ for some $1\leq j\leq k'$, and (ii) for any $1\leq i<j\leq k'$, neither $a_{i}'$ is in the subtree rooted at $a_{j}'$ nor $a_{j}'$ is in the subtree rooted at $a_{i}'$. By condition (ii), the sets $C_{a_{j}'}$, $1\leq j\leq k'$, are pairwise disjoint. Therefore we can apply Lemma~\ref{lem:hypercube-intersecting} to bound $k'$. Let $\ell'=(2\gamma+\frac{8}{\sigma+4})\ell$, according to Lemma~\ref{lem:head-size}, we have 
     \begin{align*}
         L(s(\pi(a_{j}')))\geq \frac{1}{3^dc\sqrt{d}(\sigma+4)}\ell=\frac{1}{3^dc\sqrt{d}(\sigma+4)(2\gamma+\frac{8}{\sigma+4})}\ell'=\frac{1}{2c3^d\sqrt{d}(\gamma(\sigma+4)+4)}\ell'.
     \end{align*}
     Applying Lemma~\ref{lem:hypercube-intersecting}, with $\ell=\ell'$ and $\alpha=2c3^d\sqrt{d}(\gamma(\sigma+4)+4)$, we have 
     \begin{align*}
         k'\leq c(4c3^d\sqrt{d}(\gamma(\sigma+4)+4)+2)^d.
     \end{align*}

     By condition (i), for each $j$ with $1\leq j\leq k'$, define 
     \begin{align*}
         \mathcal{H}_j=\{a_i:1\leq i\leq k\text{ and }a_j'\text{ is ancestor of }a_i\},
     \end{align*}
     and
     \begin{align*}
         \mathcal{H}_{j,l}=\{a_i:1\leq i\leq k\text{ and }a_j'\text{ is the }l\text{-th ancestor of }a_i\}.
     \end{align*}
     
     Let $h=2+d+\log (4c\sqrt{d}) + (1+d+\log (4c\sqrt{d}))\log \left(\frac{2(\sigma+4)}{\sigma+1}\right)$. According to Lemma~\ref{lem:edges-on-path-between-heads}, $\mathcal{H}_{j,l}$ is empty when $l> h$. Since the sets $C_{a_j'}$, $1\leq j\leq k'$, are pairwise disjoint, the sets $C_{a'}$, where $a'\in \bigcup\limits_{1\leq j\leq k'}\mathcal{H}_{j,l}$, are pairwise disjoint. Thus for each $l$ with $1\leq l\leq h$, applying  Lemma~\ref{lem:hypercube-intersecting} with $\ell=\ell'$ and $\alpha=2c3^d\sqrt{d}(\gamma(\sigma+4)+4)$, we have 
     \begin{align*}
         |\bigcup\limits_{1\leq j\leq k'}\mathcal{H}_{j,l}|\leq c(4c3^d\sqrt{d}(\gamma(\sigma+4)+4)+2)^d.
     \end{align*}
     Since $\mathcal{H}_j=\bigcup\limits_l\mathcal{H}_{j,l}$, we have 
     \begin{align*}
         |\bigcup\limits_{1\leq j\leq k'}\mathcal{H}_j|\leq ch(4c3^d\sqrt{d}(\gamma(\sigma+4)+4)+2)^d.
     \end{align*}

     Therefore $k\leq k'+|\bigcup\limits_{1\leq j\leq k'}\mathcal{H}_j|\leq c(h+1)(4c3^d\sqrt{d}(\gamma(\sigma+4)+4)+2)^d$. Since the dimension $d$ is constant and $c,\sigma>1$, we have $k=O(c^{d+1}(1+\gamma\sigma)^d)$ and $|\overline{\mathcal{D}}|=O(c^{2d+2}\sigma^d(1+\gamma\sigma)^d)$. The same upper bound holds for the set of dumbbells $D(a,b)$ such that at least one head of $D(a,b)$ is within graph distance $\gamma\ell$ to $head(v)$ of $D(u,v)$. This concludes the proof of the lemma.
\end{proof}

\subsection{Partition the dumbbells}\label{ssec:partition-dumbbells}
Let $\mathcal{D}$ denote the set of $m$ dumbbells constructed from the WSPD$_G$ in Corollary~\ref{thm:GWSPD-graph-distance}. Dumbbells in $\mathcal{D}$ are partitioned into a number of groups such that each group satisfies two properties. Let $D$ and $D'$ be any two dumbbells in the same group, and let $\ell$ and $\ell'$ be their length, respectively. Without loss of generality, assume that $\ell\leq \ell'$. $D$ and $D'$ are required to satisfy:\\
\textbf{Length-grouping property:} Either $\ell'\leq 2\ell$ or $\ell'\geq \sigma\ell$ where $\sigma$ is the separation ratio of the graph WSPD. Therefore, either $D$ and $D'$ has similar length, or their length differs significantly by a factor of at least $\sigma$.\\
\textbf{The empty-region property:} If $\ell'\leq 2\ell$, the graph distance between any head of $D$ and any head of $D'$ is at least $\gamma\ell$, where $\gamma$ is a positive real number. Therefore there is no (head of other) dumbbells of similar length in the vicinity of $D$'s ($D'$'s) heads.

The length-grouping property is established as follows. Sort the dumbbells $\mathcal{D}$ according to their length. Let the normalized length of the minimum dumbbell length be 1. Consider the geometric sequence of common ratio 2 starting from 1. Every two consecutive numbers in the sequence corresponds to an interval $[2^{i},2^{i+1}]$, $i\geq 0$, which we call bucket $i$. See Figure~\ref{fig:length-grouping} for an illustration. Put the dumbbells by their normalized length into the corresponding buckets in order\footnote{This requires the word RAM model of computation.}. Dumbbells in buckets $i, \lceil \log \sigma\rceil+1+i, 2\lceil \log \sigma\rceil+2+i, \ldots$, $0\leq i \leq \lceil \log \sigma\rceil$, form a dumbbell group $i$. Each dumbbell group $i$, $0\leq i \leq \lceil \log \sigma\rceil$, satisfies the length-grouping property. The running time of the algorithm is dominated by sorting the dumbbells by their length, which is $O(m\log m)$. We have the following lemma. 
\begin{figure}[H]
	\centering
		\includegraphics[width=.8\textwidth]{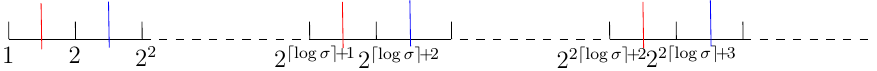}
		\caption{The buckets $0,\lceil \log \sigma \rceil +1,2\lceil \log \sigma\rceil+2,\ldots$ are stabbed with red segments. The buckets $1,\lceil \log \sigma \rceil +2,2\lceil \log \sigma\rceil+3,\ldots$ are stabbed with blue segments.}
		\label{fig:length-grouping}
\end{figure}
\begin{lemma}\label{lem:length-grouping}
    In $O(m\log m )$ time one can partition the dumbbells $\mathcal{D}$ into $\lceil \log \sigma \rceil+1$ groups. Let $D$ and $D'$ be two dumbbells in dumbbell group $i$, $0\leq i\leq \lceil \log \sigma \rceil$. Let $\ell$ ($\ell'$) be the length of $D$ ($D'$) where $\ell\leq \ell'$. Then $\ell'\leq 2\ell$ or $\ell'\geq \sigma\ell$.
\end{lemma}

For each dumbbell group $i$, let $\mathcal{E}_1,\ldots,\mathcal{E}_k$ be its dumbbells where subset $\mathcal{E}_j$, $1\leq j\leq k$, contains dumbbells of similar length. Assume that the dumbbells in $\mathcal{E}_j$ have length in the interval $[\ell,2\ell]$ for $\ell>0$. We prove in Lemma~\ref{lem:main-packing-tc} that for any dumbbell $D$ in $\mathcal{E}_j$, there are at most $t_{\sigma\gamma}=O(c^{2d+2}\sigma^d(1+\gamma\sigma)^d)$ dumbbells in $\mathcal{E}_j$ such that one head of the dumbbells is within graph distance $\gamma \ell$ to one head of $D$. 

Define a graph $G_j$ whose vertex set is the set $\mathcal{E}_j$ of dumbbells. For two dumbbells $D$ and $D'$ in $\mathcal{E}_j$, there is an edge in $G_j$ joining $D$ and $D'$ if one head of $D$ is within graph distance $2\gamma\ell$ to one head of $D'$. According to Lemma~\ref{lem:main-packing-tc}, each vertex in $G_j$ has degree at most $t_{\sigma,2\gamma}$. Thus the vertices of $G_j$ can be colored using at most $1+t_{\sigma,2\gamma}$ colors, such that no two adjacent vertices have the same color. Given $G_j$, such a coloring can be computed in $O(t_{\sigma,2\gamma}|\mathcal{E}_j|)$ time by a greedy algorithm. Dumbbells with the same color are put into the same subset, which satisfies the empty-region property. Let $\mathcal{E}_{j,l}$, $1\leq l\leq 1+t_{\sigma,2\gamma}$, be such a subset. The subsets $\mathcal{E}_{j,l}$, $1\leq l\leq 1+t_{\sigma,2\gamma}$, form a partition of $\mathcal{E}_j$.  

The graph $G_j$ is built as follows. We construct a graph $G_j'$, from which we can obtain $G_j$. Let $D(u,v)$ be a dumbbell in $\mathcal{E}_j$ and consider one of its heads, say $head(u)$. Let $D(a,b)$ be a dumbbell in $\mathcal{E}_j$ such that one of its head, say $head(a)$, is within graph distance $2\gamma\ell$ to $head(u)$. According to Lemma~\ref{lem:head-size-to-length}, $\gdub{C_a}\leq \frac{2\ell}{\sigma+1}$ and $\gdub{C_u}\leq \frac{2\ell}{\sigma+1}$. If $head(a)$ is within graph distance $2\gamma\ell$ to $head(u)$, the graph distance between the center of $head(u)$ ($rep_u$) and the center of $head(a)$ ($rep_a$) is at most 
\begin{align*}
    \frac{4\ell}{\sigma+1}+2\gamma\ell=(\frac{4}{\sigma+1}+2\gamma)\ell.
\end{align*}
Let $\mathcal{E}'$ denote the set of such $D(a,b)$. Applying Lemma~\ref{lem:main-packing-tc} with $\gamma$ replaced by $\frac{4}{\sigma+1}+2\gamma$, we have that \begin{align*}
    |\mathcal{E}'|=O(c^{2d+2}\sigma^d(1+\gamma\sigma)^d).
\end{align*}

Let the head centers of the dumbbells in $\mathcal{E}_j$ be the vertex set of $G_j'$. Construct a bounded degree 2-spanner $G_j'$ on its vertex set with respect to the graph distance. We use the $(1+\varepsilon)$-spanner with bounded degree for doubling metrics by Gottlieb and Roditty~\cite{dynamic-spanner-doubling-metrics_GottliebRoddity-esa08} to construct $G_j'$. Then for each vertex $u'$ of $G_j'$, run Dijkstra's algorithm from $u'$ and find all vertices $v$ such that $dist_{G_j'}(u',v)\leq 2(\frac{4}{\sigma+1}+2\gamma)\ell$. Check all such vertices $v$ and report all vertices $v'$ such that $dist_G(u',v')\leq (\frac{4}{\sigma+1}+2\gamma)\ell$. The graph $G_j$ is thus constructed.

The spanner graph $G_j'$ has degree $c^{O(1)}$, has $2|\mathcal{E}_j|$ vertices and can be constructed in $O(c^{O(1)}|\mathcal{E}_j|\log |\mathcal{E}_j|)$ time. Checking vertex $v$ for $dist_G(u',v)$ takes $O(c\log n)$ time by Theorem~\ref{thm:exact-distance-oracle}. The graph $G_j$ is thus constructed in time 
\begin{align*}
O(c^{O(1)}|\mathcal{E}_j|\log|\mathcal{E}_j|+|\mathcal{E}_j|c^{2d+3}\sigma^d(1+\gamma\sigma)^d\log n).
\end{align*}
We have obtained the following lemma.
\begin{lemma}\label{lem:empty-region}
    Let $\ell$ be a positive number, and let $\mathcal{E}_j$ be a set of dumbbells whose length are in the interval $[\ell,2\ell]$. In 
    \begin{align*}
        O(c^{O(1)}|\mathcal{E}_j|\log |\mathcal{E}_j|+|\mathcal{E}_j|c^{2d+3}\sigma^d(1+\gamma\sigma)^d\log n)
    \end{align*} 
    time, one can partition $\mathcal{E}_j$ into $O(c^{2d+2}\sigma^d(1+\gamma\sigma)^d)$ subsets, each of which satisfies the $\gamma$-empty-region property.
\end{lemma}

Let $\mathcal{E}_{jl}$, $1\leq l\leq O(c^{2d+2}\sigma^d(1+\gamma\sigma)^d)$, be the partitioned subsets of $\mathcal{E}_j$. The dumbbell group $i$ is partitioned into (sub)groups 
\begin{align*}
    \mathcal{D}_{l}^i =\bigcup\limits_{1\leq j\leq k}\mathcal{E}_{jl}.
\end{align*}
Each group $\mathcal{D}_{l}^i$ satisfies the length-grouping property and the $\gamma$-empty-region property. Putting Lemma~\ref{lem:length-grouping} and Lemma~\ref{lem:empty-region} together, we obtain the following result.
\begin{lemma}\label{lem:dumbbell-partition-main}
    Let $\mathcal{D}$ be the set of all $m$ dumbbells, and let $\gamma$ be a positive number. In \begin{align*}
        O(c^{O(1)}m\log m + mc^{2d+3}\sigma^d(1+\gamma\sigma)^d\log n)
    \end{align*} 
    time, we can partition $\mathcal{D}$ into subsets $\mathcal{D}_1,\mathcal{D}_2,\ldots,\mathcal{D}_l$, where 
    \begin{align*}
        l=O(c^{2d+2}\sigma^d(1+\gamma\sigma)^d\log \sigma),
    \end{align*}
    such that each subset $\mathcal{D}_j$, $1\leq j\leq l$, satisfies the length-grouping property and the $\gamma$-empty-region property.
\end{lemma}

\subsection{Construct the dumbbell trees}\label{ssec:build-dumbbell-trees}
For each dumbbell group $\mathcal{D}_j$ in Lemma~\ref{lem:dumbbell-partition-main}, we will construct a ``dumbbell tree'' $DT_j$ connecting the dumbbells in $\mathcal{D}_j$ and points in $V(G)$ hierarchically. 

Let $s_0$ be a bounding hypercube of points in $V(G)$. Let $rep$ be an arbitrary point in $V(G)$ and let $tl$ be the total length of all edges in $G$. Let $head_0$ be a graph metric ball which is centered at $rep$ and has radius $tl$. Define a dummy dumbbell $D_0$ which has $head_0$ and a translated copy of $head_0$ as its heads, and whose length is $\sigma$ times the maximum length of all dumbbells in $\mathcal{D}$.

The dumbbell tree $DT_j$ is a rooted tree consisting of \emph{dumbbell nodes}, \emph{head nodes} and \emph{leaves}.
\begin{enumerate}
    \item Each dumbbell in $\mathcal{D}_j\cup \{D_0\}$ corresponds to a dumbbell node in $DT_j$. This dumbbell node has two children, each of which is a head node corresponding to one of its heads.
    \item The dumbbell node corresponding to $D_0$ is the root of $DT_j$.
    \item For each dumbbell $D$ in $\mathcal{D}_j$, let $D'$ be the dumbbell satisfying the two conditions in Lemma~\ref{lem:dumbbell-unique-parent} and let $head'$ be head of $D'$ such that some head of $D$ is within graph distance $\gamma'\ell$ to $head'$. The dumbbell node for $D$ is a child of the head node for $head'$.
    \item Each point $p$ in $V(G)$ corresponds to a leaf of $DT_j$ which stores $p$. Let $b$ be the deepest node in the connected tree $T$ such that $p\in C_b$ and $C_b$ is the head set of some dumbbell in $\mathcal{D}_j\cup \{D_0\}$. The leaf storing $p$ is a child of the head node corresponding to $head(b)$.
    \item Each node $v$ of $DT_j$ stores a point of $V(G)$, called the \emph{representative} of $v$. 
    \begin{itemize}
        \item If $v$ is a leaf, the representative is the point stored at $v$.
        \item If $v$ is a head node, the representative is an arbitrary point in $C_a$ where $v$ is the head node for $head(a)$.
        \item If $v$ is a dumbbell node, the representative is an arbitrary point in $C_a$ or $C_b$ where $head(a)$ and $head(b)$ are the two heads of the dumbbell corresponding to $v$.
    \end{itemize}
\end{enumerate}

The following lemma statesd that the size of $DT_j$ is $O(n)$.
\begin{lemma}\label{lem:dumbbell-tree-size-linear}
    The size of the dumbbell tree $DT_j$, $1\leq j\leq l$, is $O(n)$.
\end{lemma}
\begin{proof}
    By construction, all the leaves of $DT_j$ corresponds to points in $V(G)$. Internal nodes of $DT_j$ correspond to dumbbells or heads of dumbbells. Each dumbbell node has two child head nodes. A head node has at least one child, which is either a dumbbell node or a leaf. It follows that the number of internal nodes in $O(n)$. The lemma thus holds. 
\end{proof}

By the length-grouping property and the empty-region property, the parent node of a leaf in $DT_j$ is unique. Let $\gamma'>0$ be a real number such that $\gamma'\leq \gamma$ and $\sigma\gamma\geq 2\gamma'+\frac{\sigma+3}{\sigma+1}$. The following lemma shows that for each dumbbell $D$ in $\mathcal{D}_j$, the parent node of $D$'s dumbbell node has a unique parent dumbbell node.  
\begin{lemma}\label{lem:dumbbell-unique-parent}
    Let $D$ be a dumbbell in $\mathcal{D}_j$ with length $\ell$. There is a unique dumbbell $D'$ of minimum length in $\mathcal{D}_j\cup \{D_0\}$ such that (i) the length of $D'$ is greater than $\ell$ and (ii) at least one head of $D'$ is within graph distance $\gamma'\ell$ to some head of $D$. Let $\ell'$ be the length of $D'$, then $\ell'\geq \sigma \ell$.
\end{lemma}
\begin{proof}
    Observe that the dummy dumbbell $D_0$ satisfies conditions (i) and (ii), thus the dumbbell $D'$ exists. Let $\ell'$ be the length of $D'$. From condition (i), $\ell'>\ell$. If $\ell'\leq 2\ell$, since $\gamma'\leq \gamma$, $D$ and $D'$ do not satisfy the empty-region property, a contradiction. Thus $\ell'\geq \sigma\ell$.

    Assume there is another dumbbell $D''$ of minimum length in $\mathcal{D}_j\cup \{D_0\}$ that satisfies conditions (i) and (ii). According to Lemma~\ref{lem:head-size-to-length}, both heads of $D$ has radius at most $\frac{\ell}{\sigma+1}$. Let $head'$ be a head of $D'$ which is within graph distance $\gamma'\ell$ to some head of $D$, and let $head''$ be a head of $D''$ which is within graph distance $\gamma'\ell$ to some head of $D$. By the triangle inequality, the graph distance between $head'$ and $head''$ is at most 
    \begin{align*}
        2\gamma'\ell + \frac{2\ell}{\sigma+1} + \ell=(2\gamma'+\frac{\sigma+3}{\sigma+1})\ell\leq \sigma\gamma\ell\leq \gamma\ell'.
    \end{align*}
    This implies that $D'$ and $D''$ does not satisfy the $\gamma$-empty-region property, a contradiction. Therefore $D'$ is unique.
\end{proof}

To construct the dumbbell tree $DT_j$, what remains to be done are: (i) compute the parent head node for each leaf which stores a point in $V(G)$, (ii) compute the parent head node for each dumbbell node. Task (i) is simple. Using the connected tree $T$, the parents of the leaves in $DT_j$ can be computed in $O(n)$ total time by a traversal of $T$ (recall that $T$ has $O(n)$ nodes). 

Task (ii) is harder. Let $D=D(a,b)$ be a dumbbell in $DT_j$ with length $\ell$. We need to find the dumbbell $D'$ in $DT_j$ with minimum length such that (i) the length of $D'$ is at least $\sigma\ell$, (ii) some head of $D'$ is within graph distance $\gamma'\ell$ to some head of $D(a,b)$. Let $D'=D(u,v)$ where $u,v$ are nodes in $T$. We only explain the case when $head(a)$ is within graph distance $\gamma'\ell$ to some head of $D'$. Without loss of generality, assume that $head(a)$ is close to $head(u)$ of $D'$. Consider the hypercube $\mathbf{s}$ that is centered at the center of $head(a)$ and has side length $(\frac{2}{\sigma+1}+2\gamma')\ell$. Observe that $s(u)$ intersects the hypercube $\mathbf{s}$. Construct a grid on $\mathbf{s}$, where each cell has side length $\frac{\sigma\ell}{2c\sqrt{d}(\sigma+4)}$. There are $O(c^d(1+\gamma'\sigma)^d/\sigma^d)$ grid cells. 

Similar to the split tree in~\cite{DBLP:journals/jacm/CallahanK95}, when splitting cells of the quadtree $Q$, let $s_0(u)$ be a hyperrectangle which contains the (quadtree) cell of $u$. The $s_0$ of the root of $Q$ is the bounding cube of $V(G)$. When splitting the cell of a node in $Q$, say $u$, $s_0(u)$ is partitioned into the $s_0$s of its children nodes. See Figure~\ref{fig:R0-rectangle} for an illustration. It follows that the $s_0$ rectangles of the leaves of $Q$ partitions the bounding cube of $V(G)$.

    \begin{figure}[H]
	\centering
		\includegraphics[width=.8\textwidth]{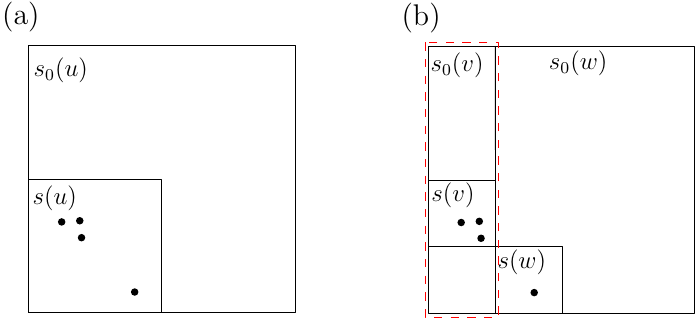}
		\caption{(a) Illustration for the $s_0$ rectangle of a node $u$. (b) When splitting the cell ($s(u)$) of $u$, $s_0(u)$ is partitioned into $s_0(v)$ and $s_0(w)$ where $s(v)$ and $s(w)$ are the non-empty children cells of $s(u)$.}
		\label{fig:R0-rectangle}
    \end{figure}

We have $L_{min}(s_0(u))\geq \frac{1}{2}\cdot L(s(\pi(u)))$, where $L_{min}(s_0(u))$ is the length of the shortest side of $s_0(u)$. According to Lemma~\ref{lem:head-size-to-length}, 
\begin{align*}
    L_{min}(s_0(u))\geq \frac{1}{2}\cdot L(s(\pi(u)))\geq \frac{\sigma\ell}{2c\sqrt{d}(\sigma+4)}.
\end{align*} 
Since $s(u)\subseteq s_0(u)$, $s_0(u)$ intersects $\mathbf{s}$ and contains at least one grid point. Let $z$ be the grid point and let $w$ be the leaf of $Q$ such that $s_0(w)$ contains $z$. Therefore, $u$ is the deepest node in the connected tree $T$ such that (i) $head(u)$ is the head of some dumbbell $D'$ in $DT_j$, (ii) the length of $D'$ is at least $\sigma\ell$, and (iii) the graph distance between $head(a)$ and $head(u)$ is at most $\gamma'\ell$.

The node $u$ and the dumbbell $D'=D(u,v)$ can be found in two steps:
\begin{enumerate}
    \item For each grid point $z$ of $\mathcal{R}$, compute the leaf $w$ in $Q$ such that $s_0(w)$ contains $z$.\label{search-step-1}
    \item For each leaf $w'$ in $T$ whose cell is $s(w)$, compute the deepest node $u$ in $T$ on the path from the root to $w'$ that satisfies (i), (ii) and (iii) above. \label{search-step-2}
\end{enumerate}

In Step~\ref{search-step-1}, for each grid point, computing the leaf $w$ takes $O(\log n)$ time by using a centroid decomposition of $Q$. In Step~\ref{search-step-2}, for each leaf $w'$, computing the deepest node $u$ takes $O(\log n)$ time by using the heavy path decomposition of $T$. For Step~\ref{search-step-2}, observe that the length of dumbbells in $DT_j$ having $head(u)$ as one head are increasing by a factor of at least $\sigma$, due to the empty-region property. Therefore given any dumbbell $D$ in $\mathcal{D}_j$, its parent in $DT_j$ can be computed in $O(c^{d+1}(1+\gamma'\sigma)^d\log n/\sigma^d)$ time. 

Since the parents for the leaves can be computed in $O(n)$ time, the dumbbell tree $DT_j$ can be constructed in $O(n\log n+c^{d+1}(1+\gamma'\sigma)^d|\mathcal{D}_j|\log n/\sigma^d)$ time. Summing up, we have proved the following lemma.
\begin{lemma}\label{lem:dumbbell-trees-construct}
    Let $\gamma,\gamma'$ be positive numbers where $\gamma'\leq \gamma$ and $\sigma\gamma\geq 2\gamma'+\frac{\sigma+3}{\sigma+1}$. The set $\mathcal{D}$ of all dumbbells is partitioned into subsets $\mathcal{D}_j$, $1\leq j\leq l$, where $l=O(c^{2d+2}\sigma^d(1+\gamma\sigma)^d\log \sigma)$. One can construct the dumbbell trees $DT_j$, $1\leq j\leq l$, in $O(nl\log n + c^{d+1}(1+\gamma'\sigma)^dm\log n/\sigma^d)$ time.
\end{lemma}

\subsection{The dumbbell trees constitute a tree cover}\label{ssec:dumbbell-trees-constitute-tc}
Let $DT_j$ be a dumbbell tree, and let $u$ and $v$ be two nodes in $DT_j$. Let $u=u_1,u_2,\ldots,u_f=v$ be the path in $DT_j$ between $u$ and $v$, and let $p_i$, $1\leq i\leq f$, be the representative of $C_{u_i}$. Let $\Pi=p_1,\ldots,p_f$ be the path where the distance between $p_i$ and $p_{i+1}$, $1\leq i < f$, equals $dist_G(p_i,p_{i+1})$. We call $\Pi$ the \emph{graph metric path between $u$ and $v$}. We can prove the following lemma. The proof is omitted here.
\begin{lemma}\label{lem:graph-metric-path-length}
    Let $u$ and $v$ be nodes in $DT_j$, where $u$ is a leaf, $v$ is a head node and $v$ is an ancestor of $u$. Let $\ell$ be the length of the dumbbell corresponding to the parent of $v$. The length of the graph metric path between $u$ and $v$ is at most $\frac{15}{\sigma}\cdot \ell$. 
\end{lemma}
Let $p$ be a point in $V(G)$, let $u$ be a leaf in $DT_j$ which stores $p$, and let $v$ be a head node in $DT_j$ such that $v$'s head set contains $p$. By setting $\sigma$ and $\gamma'$, one can guarantee that $v$ is an ancestor of $u$. We have the following lemma. The proof is omitted.
\begin{lemma}\label{lem:proper-containment}
    Assume that $\sigma\geq 3$, $\frac{15}{\sigma}\leq \gamma'$ and $\frac{15}{\sigma}<\frac{\sigma}{\sigma+2}$. Let $p$ be a point in $V(G)$, let $u$ be the leaf in $DT_j$ which stores $p$, and let $v$ be a head node in $DT_j$ such that $v$'s head set contains $p$. Then $v$ is an ancestor of $u$.
\end{lemma}
Let $D=D(v,v')$ be dumbbell in $DT_j$. Let $p$ and $q$ be two points in $V(G)$ such that $p\in C_v,q\in C_{v'}$ or $p\in C_{v'},q\in C_v$. Let $u$ and $u'$ be the leaves in $DT_j$ which stores $p$ and $q$, respectively. The following lemma shows that given any $0<\varepsilon<1$, the length of the graph metric path between $u$ and $u'$ in $DT_j$ is at most $(1+\varepsilon)dist_G(p,q)$, by setting $\sigma$ and $\gamma'$ appropriately.
\begin{lemma}
    Let $D=D(v,v')$ be dumbbell in $DT_j$. Let $p$ and $q$ be two points in $V(G)$ such that $p\in C_v,q\in C_{v'}$ or $p\in C_{v'},q\in C_v$. Let $u$ and $u'$ be the leaves in $DT_j$ which stores $p$ and $q$, respectively. Given any $0<\varepsilon<1$, by setting $\gamma'=\frac{15}{\sigma}$ and $\sigma=\frac{63}{\varepsilon}$, the length of the graph metric path between $u$ and $u'$ in $DT_j$ is at most $(1+\varepsilon)dist_G(p,q)$.
\end{lemma}
\begin{proof}
    Given any $0<\varepsilon<1$, set $\gamma'=\frac{15}{\sigma}$ and $\sigma=\frac{63}{\varepsilon}$. Then Lemma~\ref{lem:proper-containment} is satisfied. Let $w$ be the dumbbell node in $DT_j$ corresponding to $D(v,v')$. Without loss of generality, assume that $p\in C_v,q\in C_{v'}$. According to Lemma~\ref{lem:proper-containment}, $v$ is an ancestor of $u$ and $v'$ is an ancestor of $u'$.

    Let $\ell$ be the length of $D(v,v')$, let $x$ be the representative of $v$, let $y$ be the representative of $w$ and let $z$ be the representative of $v'$. The graph metric path between $u$ and $u'$ consists of the graph metric path from $u$ to $v$, the graph metric path from $v$ to $v'$, and the graph metric path from $v'$ to $u'$. Using Lemma~\ref{lem:head-size}, the length of the graph metric path between $v$ and $v'$ is
    \begin{align*}
        dist_G(x,y)+dist_G(y,z)\leq dist_G(p,q)/\sigma+ (1+2/\sigma)dist_G(p,q)=(1+3/\sigma)dist_G(p,q).
    \end{align*}
    According to Lemma~\ref{lem:graph-metric-path-length}, the graph metric path from $u$ to $v$ and the graph metric path from $v'$ to $u'$ both have length at most $\frac{15}{\sigma}\ell$. Also we have $\ell\leq (1+2/\sigma)dist_G(p,q)$. Therefore the length of the graph metric path between $u$ and $u'$ is at most 
    \begin{align*}
        \frac{30}{\sigma}(1+2/\sigma)dist_G(p,q)+(1+3/\sigma)dist_G(p,q)<(1+\frac{63}{\sigma})dist_G(p,q)=(1+\varepsilon)dist_G(p,q).
    \end{align*}
\end{proof}

Set $\gamma=\gamma'=\frac{15}{\sigma}$. Putting Corollary~\ref{thm:GWSPD-graph-distance}, Lemma~\ref{lem:dumbbell-partition-main} and Lemma~\ref{lem:dumbbell-trees-construct} together, we obtain the following theorem.
\treecovermain*
\section{Conclusion} 

We proved that for constant dimension~$d$, any $c$-packed graph admits an $O(c^3 n)$-size well separated pair decomposition, an $O(c)$-size balanced separator, an $O(c^2 n \log n)$-size exact distance oracle and an $O(c^{d+2} n)$-size approximate distance oracle. The preprocessing times are near-linear and the query times (of the distance oracles) are logarithmic.

We conclude with two open problems. Can we improve the query times to constant time? Can we construct an approximate distance oracle for low density or lanky graphs? In the literature, $c$-packed graphs have often been used as a first step to obtaining good bounds for $\lambda$-low density graphs or lanky graphs.

\bibliographystyle{plainurl}
\bibliography{reference.bib}
\end{document}